\RequirePackage{fix-cm}
\documentclass[smallextended]{svjour3}       % onecolumn (second format)
\smartqed  % flush right qed marks, e.g. at end of proof
\usepackage{graphicx}
 \usepackage{mathptmx}      % use Times fonts if available on your TeX system
\usepackage{amsmath,eucal,graphicx}
\usepackage{amssymb}
\usepackage{fullpage}
\usepackage{mathtools}
\usepackage{epstopdf}
\usepackage[normalem]{ulem}
\usepackage{float}
\usepackage{wrapfig}
\usepackage{verbatim}
\usepackage{color}

\usepackage{amsthm}
\usepackage{caption}
\usepackage{subcaption}
\usepackage{hyperref}
\usepackage[numbers]{natbib}
\usepackage{enumerate}
\mathtoolsset{showonlyrefs=true}
\newcommand{\genclone}{\mathcal{Y}_t}
\newcommand{\gentotal}{\mathcal{B}_t}
\newcommand{\pr}{\mathbb{P}}
\newcommand{\Fze}[3]{\mathop{F\/}\nolimits\!\left(\genfrac{}{}{0pt}{0}{#1}{#2};#3\right)}

\let\alp\alpha
\spdefaulttheorem{assumption}{Assumption}{}{}
\newcommand{\LTL}{\mathfrak{L}}

\DeclareMathOperator*{\argmax}{arg\,max}

\mathtoolsset{showonlyrefs=true}

\newtheorem{prop}{Proposition}
\theoremstyle{definition}
   {
      \theoremstyle{plain}
      \newtheorem{theorema}{Theorem}[section]
  }
   {
      \theoremstyle{plain}
      \newtheorem{propa}{Proposition}[section]
  }
     {
      \theoremstyle{plain}
      \newtheorem{lemmaa}{Lemma}[section]
  }

\begin{document}

\title{Universal asymptotic clone size distribution for general population growth%\thanks{Grants or other notes
}


\author{Michael D. Nicholson         \and
       Tibor Antal %etc.
}


\institute{M.D. Nicholson \at
              SUPA, School of Physics and Astronomy, University of Edinburgh, Edinburgh EH9 3FD, UK\\
             % Tel.: +123-45-678910\\
             \email{Michael.Nicholson@ed.ac.uk}           %  \\
           \and
           T. Antal \at
              School of Mathematics, University of Edinburgh, Edinburgh EH9 3FD, UK \\
                \email{Tibor.Antal@ed.ac.uk}   
}

\date{Received: date / Accepted: date}

\maketitle

\begin{abstract}
Deterministically growing (wild-type) populations which seed stochastically developing mutant clones have found an expanding number of applications from microbial populations to cancer. The special case of exponential wild-type population growth, usually termed the Luria-Delbr\"uck or Lea-Coulson model, is often assumed but seldom realistic. In this article we generalise this model to different types of wild-type population growth, with mutants evolving as a birth-death branching process. Our focus is on the size distribution of clones - that is the number of progeny of a founder mutant - which can be mapped to the total number of mutants. Exact expressions are derived for exponential, power-law and logistic population growth. Additionally for a large class of population growth we prove that the long time limit of the clone size distribution has a general two-parameter form, whose tail decays as a power-law. Considering metastases in cancer as the mutant clones, upon analysing a data-set of their size distribution, we indeed find that a power-law tail is more likely than an exponential one.

\end{abstract}

\section{Introduction}
\label{intro}
Cancerous tumours spawning metastases, bacterial colonies developing antibiotic resistance or pathogens kickstarting the immune system are examples in which events in a primary population initiate a distinct, secondary population. Regardless of the scenario under consideration, the number of individuals in the secondary population, and how they are clustered into colonies, or \textit{clones}, is of paramount importance. An approach which has offered insight has been to bundle the complexities of the initiation process into a mutation rate and assume that the primary, or \textit{wild-type}, population seeding the secondary, or \textit{mutant}, population is a random event.

This method was pioneered by microbiologist Salvador Luria and theoretical physicist Max Delbr\"{u}ck \cite{Luria:1943}. In their Nobel prize winning work, they considered an exponentially growing, virus susceptible, bacterial population. Upon reproduction, with small probability, a virus resistant mutant may arise and initiate a mutant clone. This model was contrasted with each wild-type individual developing resistance upon exposure to the virus with a constant probability per individual. By considering the variance in the total number of mutants in each case they demonstrated that bacterial evolution developed spontaneously as opposed to adaptively in response to the environment.

  In the original model of Luria and Delbr\"{u}ck, both wild-type and mutant populations grow deterministically, with mutant initiation events being the sole source of randomness. Lea and Coulson \cite{Lea:1949} generalised this process by introducing stochastic mutant growth in the form of the pure birth process and were able to derive the distribution of the number of mutants for neutral mutations. This was again extended by Bartlett and later Kendall \cite{Bartlett:1955,Kendall:1960}, who considered both populations developing according to a birth process. An accessible review discussing these formulations is given by Zheng \cite{Zheng:1999}.

Recent developments have focused on cancer modelling, where usually mutant cell death is included in the models. The main quantity of interest in these studies has been the total number of mutant cells. Explicit and approximate solutions appeared for deterministic, exponential wild-type growth, corresponding to a fixed size wild-type population \cite{Angerer:2001,Dewanji:2005,Iwasa:2006,Komarova:2007,Keller:2015}, and fully stochastic wild-type growth either at fixed time or fixed size \cite{Durrett:2010,Antal:2011,Kessler:2015}. An exciting recent application has been to model emergence of resistance to cancer treatments \cite{Kessler:2014,Bozic:2013,Bozic:2014}. The current study continues in this vein with our inspiration being primary tumours (wild-type) seeding metastases (mutant clones).
 
 Interestingly, in the large time small mutation rate limit, the clone size distribution at a fixed wild-type population size coincide for stochastic and deterministic exponential wild-type growth \cite{Kessler:2015,Keller:2015}. The intuition behind this observation is that a supercritical birth-death branching process converges to exponential growth in the large time limit, and, for a small mutation rate, mutant clones are initiated at large times. So asymptotically the two methods are equivalent, but the deterministic description of the wild-type population has twofold advantages: (i) the calculations are much simpler in this case \cite{Keller:2015}, and (ii) the method can be easily generalised to arbitrary growth functions. This is the programme that we develop in the present paper.

 The present work differs from previous approaches in two ways. Firstly, motivated by populations with environmental restrictions, we move away from the assumption of exponential wild-type growth, a setting which has received limited previous consideration as discussed in \cite{Foo:2014}. We shall first review and extend results for the exponential case, and then provide explicit solutions for  power-law and logistic growth. Next we present some general results which are valid for a large class of growth functions. This extends the classic results found in \cite{Kendall:1948,Athreya:2004,Karlin:1981,Tavare:1987} and recent work in  \cite{Tomasetti:2012,Houchmandzadeh:2015} who considered the wild-type population growth rate to be time-dependent but coupled with the mutant growth rate. Secondly, rather than the total number of mutants, our primary interest is on the distribution of mutant number in the clones initiated by mutation events. This complements \cite{Hanin:2006}, which allowed deterministic wild-type and mutant growth, and the treatment of clone sizes for constant wild-type populations found in \cite{Dewanji:2011}. Whilst we focus on clone sizes, we demonstrate that the distribution for the total number of mutants follows as a consequence, and hence results hold in that setting also.

The outline of this work is as follows. We define our model in Section \ref{theory}, utilising formalism introduced in \cite{Karlin:1981}, and demonstrate a mapping between the mutant clone size distribution and the distribution for the total number of mutants. 
The exact time-dependent size distribution is given for exponential, power-law and logistic wild-type growth function in Section \ref{clonesize}. Section \ref{longtime} pertains to universal features of the clone size distribution and contains our most significant results. There, for a large class of wild-type growth functions, we demonstrate a general two parameter distribution for clone sizes at large times. The distribution has power-law tail behaviour which corroborates previous work \cite{Iwasa:2006, Durrett:2010, Williams:2016}. Large time results are also given for the mean and variance of the clone sizes under general wild-type growth. Adopting the interpretation of the wild-type population as the primary tumour and mutant clones as metastases, we test our results regarding the tail of the distribution on empirical metastatic data in Section \ref{applications}. Section \ref{alt} considers alternative methods to ours and we give some concluding remarks in Section  \ref{discussion}.

%%%%%%%%%%%%%%%%%%%%%%%%%%%%%%%%    
\section{Model}\label{theory}

In our model a wild-type population gives rise to mutants during reproduction events. The arisen mutant also reproduces and so mutant clones stem from the original initiating mutant's progeny.  In many applications, the wild-type population is significantly larger than the mutant clones and so we treat the wild-type population's growth as deterministic, with size dictated by a time-dependent function $n_{\tau}$. The mutant clones are smaller in comparison and so their growth is stochastic. For logistic wild-type growth a sample realisation of the process is shown in Figure \ref{sample}. The exact formulation is now given.
\begin{figure}[h!]
\includegraphics[width=\linewidth]{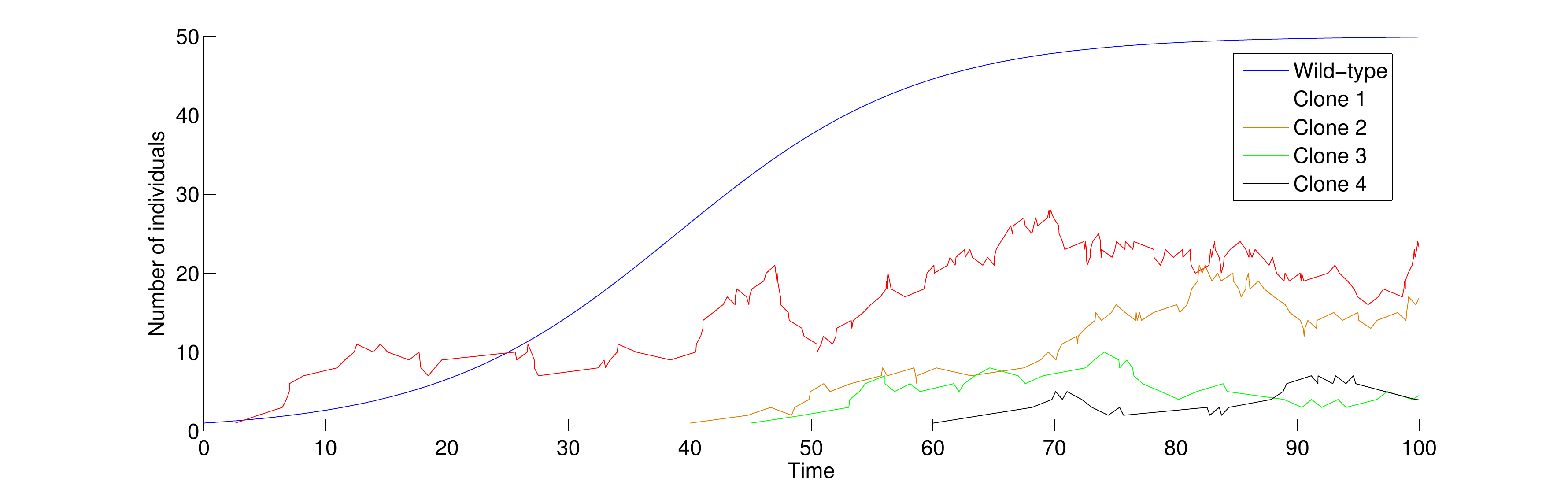}
\caption{ A sample realisation for deterministic logistic wild-type growth, with a carrying capacity of 50, and stochastic mutant growth.  Note that we typically assume the wild-type population is much larger than individual clones.}
\label{sample}
\end{figure}
  
\subsection{The birth-death process}
Stochastic growth of mutants will follow a birth-death branching process \cite{Athreya:2004}. Time is scaled such that each mutant has unit birth rate and death rate $\beta$. A brief note on converting our results to the case when the birth rate is arbitrary is given in Appendix \ref{appendtheory}. Let $Z_{t}$ be the size of a population at time $t$, with $Z_{0}=1$. The forward Kolmogorov equation for the distribution is given by
\def \alpha {}
\begin{equation}
%\partial_{t}  \mathbb{P}(Z_{t}=k)=\alpha (k-1) \mathbb{P}(Z_{t}=k-1)+\beta(k+1)  \mathbb{P}(Z_{t}=k+1)-(\alpha+\beta) k \mathbb{P}(Z_{t}=k)
\label{FK}
\partial_{t}  \mathbb{P}(Z_{t}=k)=\alpha (k-1) \mathbb{P}(Z_{t}=k-1)+\beta(k+1)  \mathbb{P}(Z_{t}=k+1)-(1+\beta) k \mathbb{P}(Z_{t}=k)
\end{equation}
with $k\geq 1$. Its solution in terms of the generating function, given on page 76 of \cite{Bartlett:1955},  is 
\begin{gather}
\label{BDgf}
\mathcal{Z}_{t}(s)=\mathbb{E}(s^{Z_{t}})=1-\frac{\lambda}{1-\xi e^{-\lambda t}}, \text{ where }\xi= \frac{\beta-s}{1-s},\,\lambda=1-\beta.
\end{gather}
Due to our timescale, $\beta$ is the probability of eventual extinction for a mutant clone for $\beta\leq1$, and $\lambda$ is the mutant fitness.  When $\beta=0$ and so the stochastic proliferation follows a pure birth or Yule process, the mutants will be denoted immortal. By expanding the generating function around $s=0$, we obtain for the probability of the population size being $k$ a geometric distribution with a modified zero term
%  \begin{gather}
%  \mathbb{P}(Z_{t}=k)=\frac{1}{k!}\lim_{s\rightarrow 0}\partial_{s}^{k}\mathcal{Z}_{t}(s).
%  \end{gather}
%\tib{needed?:}  A technical point is that for some generating functions, due to parameter choices, there may be different power series representations. As our interest is in the limit $s\rightarrow 0$ we choose the representation corresponding to small $s$. 
%\begin{equation}
%  \label{BDpmf}
%\mathbb{P}(Z_{t}=k)=
%     \begin{cases}
%     q/ \mathcal{S}_{t} & \mbox{for } k=0\\
%       \frac{(1-q)^2e^{-\lambda t}}{(1-qe^{-\lambda t})^2}\,  \mathcal{S}_{t}^{1-k}  & \mbox{for } k\geq 1.
% \end{cases}
%\end{equation}
%or:?
\begin{equation}
  \label{BDpmf}
\mathbb{P}(Z_{t}=k)=
     \begin{cases}
     \beta/ \mathcal{S}_{t} &  k=0\\
      (1-\beta/\mathcal{S}_t) (\mathcal{S}_t-1)\,  \mathcal{S}_{t}^{-k}  &  k\geq 1,
%      (1-q/S_t) (1-1/S_t)\,  (1/\mathcal{S}_{t})^{k-1}  & \mbox{for } k\geq 1.
 \end{cases}
\end{equation}
with the shorthand notation
\begin{equation}
\label{ratio}
\mathcal{S}_{t}=\frac{1-\beta e^{-\lambda t}}{1-e^{-\lambda t}}.
\end{equation}
For the particular case of a critical branching process, i.e. when birth and death rates are equal, the above probabilities are simplified by observing
\begin{equation}
\label{critsimp}
\lim_{\beta\rightarrow 1}\mathcal{S}_t=\frac{t+1}{t}.
\end{equation}

%and the usual notation $\phi=1-e^{-\lambda t}$. By calling $EZ_{t}=e^{\lambda t}=N$, we get  $\phi=1-1/N$, which in some limit formulas plays the role of a finite size corrections.

\subsection{Mutant clone size distribution}
Here we employ standard methods as outlined in, for instance, \cite{Karlin:1981,Dewanji:2005}. The system is observed at a fixed time $t$ and we let the number of wild-type individuals be denoted by $n_{\tau}$ for $0\leq \tau\leq t$. Since mutants are produced by wild-type individuals, the rate of mutant clone initiations will be proportional to the product of $n_{\tau}$ and the mutation rate $\mu$.  More precisely, the process of clone initiations is an inhomogeneous Poisson process \cite{Karlin:1998} with intensity $\mu n_{\tau}$. Let the Poisson random variable $K_{t}$ denote the number of clones that have been initiated by $t$, which has mean
\begin{gather}
\mathbb{E}(K_{t})=\int_{0}^{t}\mu n_{\tau}\,d\tau.
 \end{gather}
 Now, assuming $K_{t}>0$, we consider a mutant clone sampled uniformly from the $K_t$ initiated clones and denote its size to be the random variable $Y_{t}$. The clone was initiated at the random time $T$ and as we must have $T\leq t$, the density of $T$ is given by 
\begin{gather}
\label{timedens}
f_{T}(\tau)=\frac{\mu n_{\tau}}{\mathbb{E}(K_{t})}=\frac{n_{\tau}}{a_{t}}.
\end{gather}
Where
\begin{equation}\label{rateone}
a_{t}=\frac{\mathbb{E}(K_{t})}{\mu}=\int_{0}^{t}n_{\tau}\,d\tau
\end{equation}
is the expected number of clones seeded when the mutation rate is unity. The size of the clone is dictated not only by the initiation time but also by its manner of growth, here the birth-death process. Hence, by conditioning on the arrival time, we have
\begin{gather}
\label{Singleclonepmf}
\mathbb{P}(Y_{t}=k)=\frac{1}{a_{t}}\int_{0}^{t} n_{\tau} \mathbb{P}(Z_{t-\tau}=k)\,d\tau.
\end{gather} 
An immediate consequence is that the generating function of the clone size is given by
\begin{gather}
\label{SinglecloneGF}
\genclone(s)=\mathbb{E}(s^{Y_{t}})%=\sum_{n\geq 0}s^n \mathbb{P}(Z_{t-T}=n)
=\frac{1}{a_{t}}\int_{0}^{t} n_{\tau} \mathcal{Z}_{t-\tau}(s)\,d\tau,
\end{gather}
where $\mathcal{Z}_{t}(s)$ is the generating function of the birth-death process \eqref{BDgf}.

We make the following remarks on the above. (i) The mutation rate $\mu$ does not appear in the density for initiation times in \eqref{timedens}, hence mutant clone sizes are independent of the mutation rate and thus all following results regarding clone sizes will be also. 
(ii) The integral in \eqref{Singleclonepmf} is a convolution and as convolutions commute we may swap the arguments of the integrand functions ($n_{\tau}\mathcal{Z}_{t-\tau}\leftrightarrow n_{t-\tau}\mathcal{Z}_{\tau}$).
(iii) If we start with $n_{0}$ wild-type individuals, so the wild-type follows $m_{\tau}=n_{0}n_{\tau}$, then both the numerator and denominator in \eqref{timedens} will have a factor of $n_{0}$, which cancel. So henceforth, apart from when $n_{0}=0$ (used occasionally for analytic convenience), we set $n_{0}=1$ without loss of generality.
(iv) By similar logic, a positive random amplitude for the wild-type growth function, i.e. $m_{\tau}=Xn_{\tau}$ for a general positive random variable $X$, would also cancel and so our results on clone sizes hold in that case also.

%%%%%%%%%%%%%%%%%%%%%%%%%%%%%%%
\section{Mapping distributions: clone size to total mutant number}
\label{Connection}

This section is related to the classic Luria-Delbr\"{u}ck problem. 
Let $B_{t}$ be the total number of mutants existing at time $t$. Then $B_{t}$ is the sum of $K_{t}$ generic clones
\begin{gather}
B_{t}=\sum_{i=1}^{K_{t}}(Y_{t})_{i}\,,
\end{gather}
where all $(Y_{t})_{i}$ are \textit{iid} random variables specifying the clone sizes. 
As such, $B_{t}$ is a compound Poisson random variable, and hence its generating function is
\begin{gather}
\label{GFlink}
\gentotal(s)=\mathbb{E}(s^{B_{t}})=e^{\mathbb{E}(K_{t})[\genclone(s)-1]},
\end{gather}
which can be derived by conditioning on $K_{t}$. It follows that
  \begin{equation}\label{meanlink}
  \mathbb{E}(B_{t})=\mathbb{E}(K_{t})\mathbb{E}(Y_{t})\text{ and } \mathrm{Var}(B_{t})=\mathbb{E}(K_{t})\mathbb{E}(Y_{t}^2).
  \end{equation}
The link between the mass functions of the mutant clone size, $Y_{t}$, and the total number of mutants, $B_{t}$, is given by
\begin{equation}\label{Link3}
  \mathbb{P}(B_{t}=n)= 
     \begin{cases}
     e^{\mathbb{E}(K_{t})(\mathbb{P}(Y_{t}=0)-1)}&  n=0\\
      \mathbb{E}(K_{t})\sum_{k=0}^{n-1}\frac{n-k}{n}\mathbb{P}(B_{t}=k)\mathbb{P}(Y_{t}=n-k) &   n\geq 1.
     \end{cases}
  \end{equation}
  This relationship may be found as Lemma 2 in \cite{Zheng:1999} and a short proof is provided for convenience in Appendix \ref{appendtheory}, Lemma \ref{recursion}. Hence while we may initially work in the setting of size distribution of a single clone, by the above discussion, results are transferable to the total number of mutants case.

 Often long-time results are sought, which significantly reduces the complexity of the distributions. For any fixed positive mutation rate, in the long-time limit, an infinite number of clones will have been initiated, and thus the probability distributions of $B_{t}$ will not be tight \cite{Durrett:1996}. A common solution to this problem is the \textit{Large Population-Small Mutation} limit \cite{Keller:2015}, where $\theta=\mu n_{t}$ is kept constant. Then for exponential wild-type growth (or exponential-type, see Section \ref{longtime}) the expected number of initiated clones, $\mathbb{E}(K_t)$, tends to $\theta$ for large times. Hence we see that
 \begin{gather}
\lim_{\substack{t\rightarrow \infty\\ \theta \,\mathrm{  constant}}}\gentotal(s)=\exp\big[\theta(\lim_{t\rightarrow\infty}\genclone(s)-1)\big],
\end{gather}
demonstrating that the limit of the clone size distribution is of primary concern. Furthermore, if the expected number of initiated clones is small, we have the following proposition, whose proof can be found in Appendix \ref{appendtheory}.
\begin{prop}\label{smallarrive}
For a small expected number of initiated clones, conditioned on survival, the size of a single clone and the total number of mutants are approximately equal in distribution. That is,
\begin{gather}
\mathbb{P}(B_{t}=k|B_{t}>0) = \mathbb{P}(Y_{t}=k|Y_{t}>0)+O(\mathbb{E}(K_{t}) ),\quad  \text{as}\quad \mathbb{E}(K_{t})\rightarrow 0.
\end{gather}
\end{prop}
One immediate consequence of this result is that for immortal mutants ($\beta=0$) and $\mathbb{E}(K_{t}) \ll 1$ we have
\begin{gather}
\gentotal(s)\approx (1-e^{-\mathbb{E}(K_{t}) })\genclone(s)+e^{-\mathbb{E}(K_{t}) }\implies
\mathbb{P}(B_{t}=k)\approx \mathbb{E}(K_{t})  \mathbb{P}(Y_{t}=k) \quad \text{for}\, k\geq 1.
  \end{gather}
This agrees with intuition as for small enough $\mathbb{E}(K_{t})$, we expect only 0 or 1 clones to be initiated and hence the total number of mutants will be dictated by the clone size distribution. With exponential wild-type growth this approximation was used in \cite{Iwasa:2006} to investigate drug resistance in cancer.

%%%%%%%%%%%%%%%%%%%%%%%%%%%%%%%%  
 \section{Finite time clone size distributions}\label{clonesize}
 
    Three particular cases of wild-type growth function, $n_{\tau}$, will be considered in detail, namely: exponential, power-law and logistic. Exponential and logistic growth are widely used in biological modelling \cite{Murray:2002}. For the power-law cases, under the assumption that the radius of a spherical wild-type population is proportional to time, quadratic and cubic power-law growth represents mutation rates proportional to the surface area and volume respectively. In each case we give the generating function and probability mass function. We stress again that the mutation rate and an arbitrary positive prefactor for $n_{\tau}$ cancel in \eqref{Singleclonepmf} and so are irrelevant for our results.

\subsection{Exponential wild-type growth}
Let the wild-type population grow exponentially, that is $n_{\tau}=e^{\delta \tau}$ with $\delta>0$ and so from \eqref{rateone}, $a_{t}=\frac{e^{\delta t}-1}{\delta}$. The distribution for the total number of mutants, $B_{t}$, was treated exhaustively in \cite{Keller:2015} and we follow their notation by letting $\gamma=\delta/\lambda$. Using \eqref{GFlink} and the results found in section 3 of \cite{Keller:2015}, the generating function is
\begin{gather}
\label{expgf}
\genclone(s)=1+\frac{\lambda}{1-n_{t}^{-1}}\bigg[ n_t^{-1} \Fze{1,\gamma}{1+\gamma}{\xi n_{t}^{-1/\gamma}}-\Fze{1,\gamma}{1+\gamma}{\xi}\bigg].
\end{gather}
Similarly the mass function is
  \begin{gather}
  \label{exppmf0}
\mathbb{P}(Y_{t}=0)=1+\frac{\lambda}{1-n_{t}^{-1}}\bigg[n_t^{-1}  \Fze{1,\gamma}{1+\gamma}{\beta n_{t}^{-1/\gamma}}-\Fze{1,\gamma}{1+\gamma}{\beta}\bigg]
\end{gather}
and for $k\geq 1$
\begin{equation}
\label{exppmf}
\begin{split}
 \mathbb{P}(Y_{t}=k) &= \frac{\delta}{\alpha( n_{t}-1)}\sum_{j=1}^{k}  {k-1\choose j-1}\frac{1}{j+\gamma}\bigg(\frac{\lambda}{\beta-n_{t}^{1/\gamma}}\bigg)^{j}\Fze{1,\gamma}{1+\gamma+j}{\beta n_{t}^{-1/\gamma}}
 \\
 &+\frac{\delta}{\alpha(1-n_{t}^{-1})}\frac{(k-1)!}{(\gamma+1)_{k}}\Fze{k,\gamma}{1+\gamma+k}{\beta}.
 \end{split}
\end{equation} 
Here $\Fze{a,b}{c}{z}$ is Gauss's hypergeometric function and $(a)_{k}$ is the Pochhammer symbol defined in Appendix \ref{specialfunctions}. The above expressions are given in terms of $n_{t}$ to allow easy comparison to the formulas in \cite{Keller:2015}. For these exact time-dependent formulas, their form is somewhat cumbersome, however simpler long-time limit expressions are given in Section \ref{longtime}. A reduction in complexity is also obtained for the special case of neutral mutants ($\delta=\lambda$) where, by using \eqref{hypgeom4}, the generating function in \eqref{expgf} simplifies to
$$
 \genclone(s) = 1+ \frac{\lambda}{\xi (1-e^{-\delta t})} \log \frac{1-\xi}{1-\xi e^{-\delta t}}.
$$
If additionally the neutral mutants are immortal, the above expression further simplifies to
$$
 \genclone(s) = 1+ \frac{1-s}{s\phi} \log (1-s\phi)\text{ where }\phi = 1-e^{-\delta t}.
$$
The probabilities are then concisely given by
$$
\label{neutimmtime}
 \pr(Y_t=k) = \frac{\phi^{k-1}}{k} - \frac{\phi^k}{k+1} \,\mbox{ or }\, \pr(Y_t>k)=\frac{\phi^{k}}{k+1} 
$$
which  corresponds to the classical Lea-Coulson result \cite{Lea:1949}
$$
 \gentotal(s) = (1-s\phi)^{\theta(1-s)/s}
$$
with $\theta=\mu e^{\delta t}$. %$\theta=\frac{\mu}{\alpha} e^{\delta t}$. 

\subsection{Power-law wild-type growth}
Now we assume that the wild-type population grows according to a general power-law $n_{\tau}=\tau^{\rho}$, for some non-negative integer $\rho$, and therefore $a_{t}=\frac{t^{\rho+1}}{\rho +1}$. With power-law wild-type growth and stochastic mutant proliferation, the mutant clone size generating function is given by
 \begin{gather}
 \label{powergf}
 \genclone(s)=\beta+\lambda(\rho+1) !\bigg[\frac{(-1)^{\rho} \mathrm{Li}_{\rho+1}( \xi e^{-\lambda t})}{(t\lambda)^{\rho+1}}+\sum_{i=0}^{\rho}\frac{(-1)^{i+1} \mathrm{Li}_{i+1}(\xi)}{(\rho-i)! (t\lambda)^{i+1}}
\bigg].
 \end{gather}
 Here  $\mathrm{Li}_{i}(s)$ is the polylogarithm of order $i$ defined in Appendix \ref{specialfunctions}.  Details of the derivation are given in Appendix \ref{clonalappend}. For immortal mutants, the mass function may be explicitly written as
\begin{align}
\mathbb{P}(Y_{t}=m)=\frac{(\rho+1)}{m t \alpha}+\frac{(\rho+1)!}{mt\alpha}\bigg[\frac{(-1)^{\rho}}{(t\alpha)^{\rho}}\sum_{k=1}^{m}{m\choose k}\frac{(-e^{-\alpha t})^{k}}{k^{\rho}}
\notag
\\
+\sum_{i=1}^{\rho}\frac{(-1)^{i+1}}{(t\alpha)^{i}(\rho-i)!}\sum_{k=1}^{m}{m\choose k}\frac{(-1)^{k}}{k^{i}}
\bigg].
\label{powerpmf}
\end{align}
If mutants may die, the exact mass function is most easily obtained via Cauchy's integral formula which may be efficiently computed using the fast Fourier transform. For a brief discussion on implementation see \cite{Antal:2010} and references therein.

Note for $\rho\geq1$, $n_{0}=0$ which, while useful for analytic tractability, is unrealistic. This can be overcome by letting $n_{\tau}=n_{0}+\tau^{\rho}$. Then, by splitting the integral in the generating function \eqref{SinglecloneGF} and using the above analysis, one can obtain the mass function for any $n_{0}$. However for practical purposes the contribution of $n_{0}$ is negligible.

\subsection{Constant size wild-type}

For the specific power-law growth when $\rho=0$, i.e. $n_{\tau}=1$ (recall that this is equal to the general case when $n_{\tau}=n_{0}$), we recover some classical results for constant immigration \cite{Kendall:1948}. We note that the distribution of the ordered clone size, depending on initiation time, was discussed in \cite{Jeon:2008}. From \eqref{powergf} with $\rho=0$, the generating function is

%\begin{gather}
%\label{constgf}
%\genclone(s)=q - \frac{1}{t\alpha}\log\left(\frac{1-q}{\tib{1-s-(q-s)e^{-\lambda t}}}\right).
%\end{gather}
%or
%
\begin{gather}
\label{constgf}
\genclone(s)=1-\frac{1}{t\alpha}\log\left( \frac{1-s\mathcal{S}_t^{-1}}{1-\mathcal{S}_t^{-1}} \right).
\end{gather}
%
%Observing that, for $\left|\frac{Bs}{1-A}\right|<1$,
%\begin{gather}
%\log[1-(A+Bs)]=\log(1-A)-\sum_{k\geq 1}\frac{s^k}{k}\bigg(\frac{B}{1-A}\bigg)^k
%\end{gather}
%so with $A=qe^{-\lambda t}$ and $B=\phi$ a power series representation of the generating function is given by
%\begin{gather}
%\genclone(s)=q+\frac{1}{t\alpha}\bigg[\log\bigg(\frac{1-q}{1-qe^{-\lambda t}}\bigg)+\sum_{k\geq 1}\frac{s^k}{k}\bigg(\frac{\phi}{1-qe^{-\lambda t}}\bigg)^k\bigg].
%\end{gather}
with $\mathcal{S}_t$ as given in \eqref{ratio}. By expanding this generating function in terms of $s$ we obtain the probabilities 
 \begin{equation}
 \label{constpmf}
  \mathbb{P}(Y_{t}=k)= 
     \begin{cases}
     1+t^{-1}\log(1-\mathcal{S}_t^{-1}) &  k=0\\
     \frac{1}{t\alpha k}\,\mathcal{S}_t^{-k} &  k\geq 1.
     \end{cases}
   \end{equation}
Then using \eqref{GFlink} with the clone sizes \eqref{constgf} we obtain the generating function of the total number of mutants
$$
 \gentotal(s) = \left[ \frac{1-\mathcal{S}_{t}^{-1}}{1-s\mathcal{S}_{t}^{-1}} \right]^{\mu},
$$
and from the binomial theorem we also get the probabilities
$$
 \pr(B_t=m) =  \binom{m+\mu \alpha -1}{m}\left( 1-\mathcal{S}_{t}^{-1}\right) ^{\mu \alpha} \mathcal{S}_t^{-m}.
$$
We recognise this as a negative binomial distribution under the interpretation that $B_{t}$ is the number of failures until $\mu$ successes, with failure probability $\mathcal{S}_{t}^{-1}$. 
This result for $B_{t}$ was first derived by Kendall \cite{Kendall:1948} who was attempting to explain the appearance of the logarithmic distribution for species number when randomly sampling heterogeneous populations, conjectured by R.A. Fisher. From the distribution of $B_{t}$, by an argument which may be considered a special case of Proposition \ref{smallarrive}, he derived that for constant rate initiation, the clone size conditioned on non-extinction is logarithmically distributed again with parameter $\mathcal{S}_{t}^{-1}$, which can be obtained via \eqref{constpmf}.

Constant immigration may imply a constant size source, hence mutants with equal birth and death rates (i.e. evolving as a critical branching process) are particularly interesting.
This case yields analogous formulas to those above but $\mathcal{S}_t$ is replaced with the expression given in \eqref{critsimp}.
%$$
%\lim_{\beta\rightarrow 1}\mathcal{S}_t=\frac{t+1}{t}.
%$$

% \begin{displaymath}
%  \mathbb{P}(Y_{t}=k|Y_{t}>0)=\frac{1}{k[\lambda t-\log(\frac{1-q}{1-qe^{-\lambda t}})]}\bigg(\frac{\phi}{1-qe^{-\lambda t}}\bigg)^k.
%\end{displaymath}

\subsection{Logistic wild-type growth}
Starting from a population of one and having a carrying capacity $K$, logistic growth is given by $n_{\tau}=\frac{K e^{\lambda \tau}}{K+e^{\lambda \tau}-1}$. 
We assume neutral mutations, i.e. $\lambda$ is also the wild-type growth rate. Integrating the growth function gives $a_{t}=\frac{K}{\lambda}\log(\frac{e^{\lambda t}}{n_{t}})$.

We aim to calculate the generating function using \eqref{SinglecloneGF}. Recalling the definition of $\mathcal{Z}_{t-\tau}(s)$ we observe that
\begin{gather}
\int\frac{1}{1-\xi e^{-\lambda(t-\tau)}}n_{\tau}\,d\tau=\frac{K}{\lambda[(K-1)\xi e^{-\lambda t}+1]}\log\bigg(\frac{1-e^{\lambda \tau}-K}{1-Ae^{\lambda \tau}}\bigg)+C,
\end{gather}
where $C$ is an integration constant. Therefore the generating function is
\begin{gather}
\label{loggf}
\genclone(s)=1+\frac{\lambda e^{\lambda t}}{[e^{\lambda t}+(K-1)\xi]\log(\frac{e^{\lambda t}}{n_{t}})}\log\bigg(\frac{n_{t}(1-\xi)}{e^{\lambda t}(1-\xi e^{-\lambda t})}\bigg).
\end{gather}

\begin{figure}[t!]
\centering
\begin{subfigure}{0.46\textwidth}
\includegraphics[width=\textwidth,height=6cm]{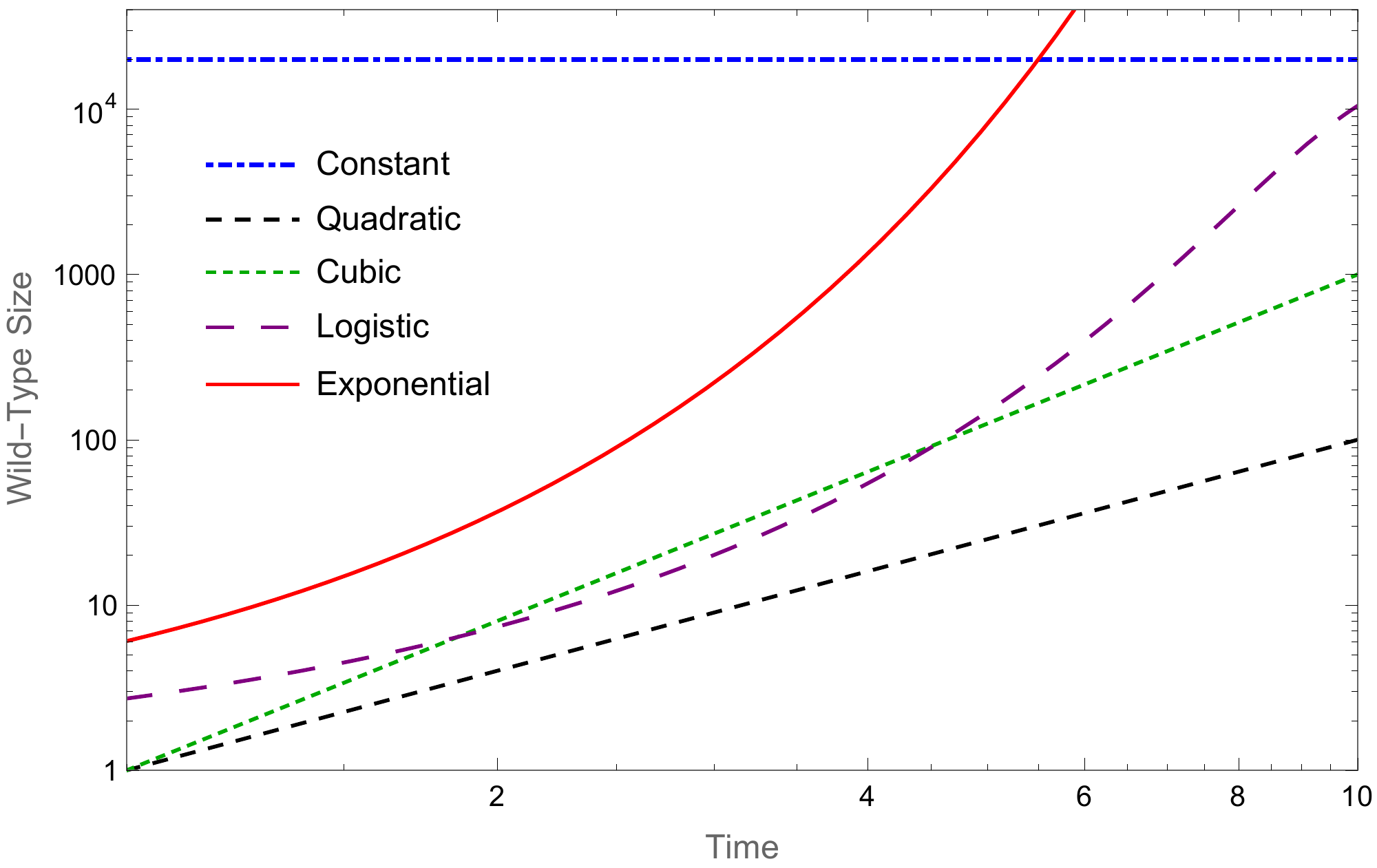}
\caption{}
\end{subfigure}
\begin{subfigure}{0.46\textwidth}
\includegraphics[width=\textwidth,height=6cm]{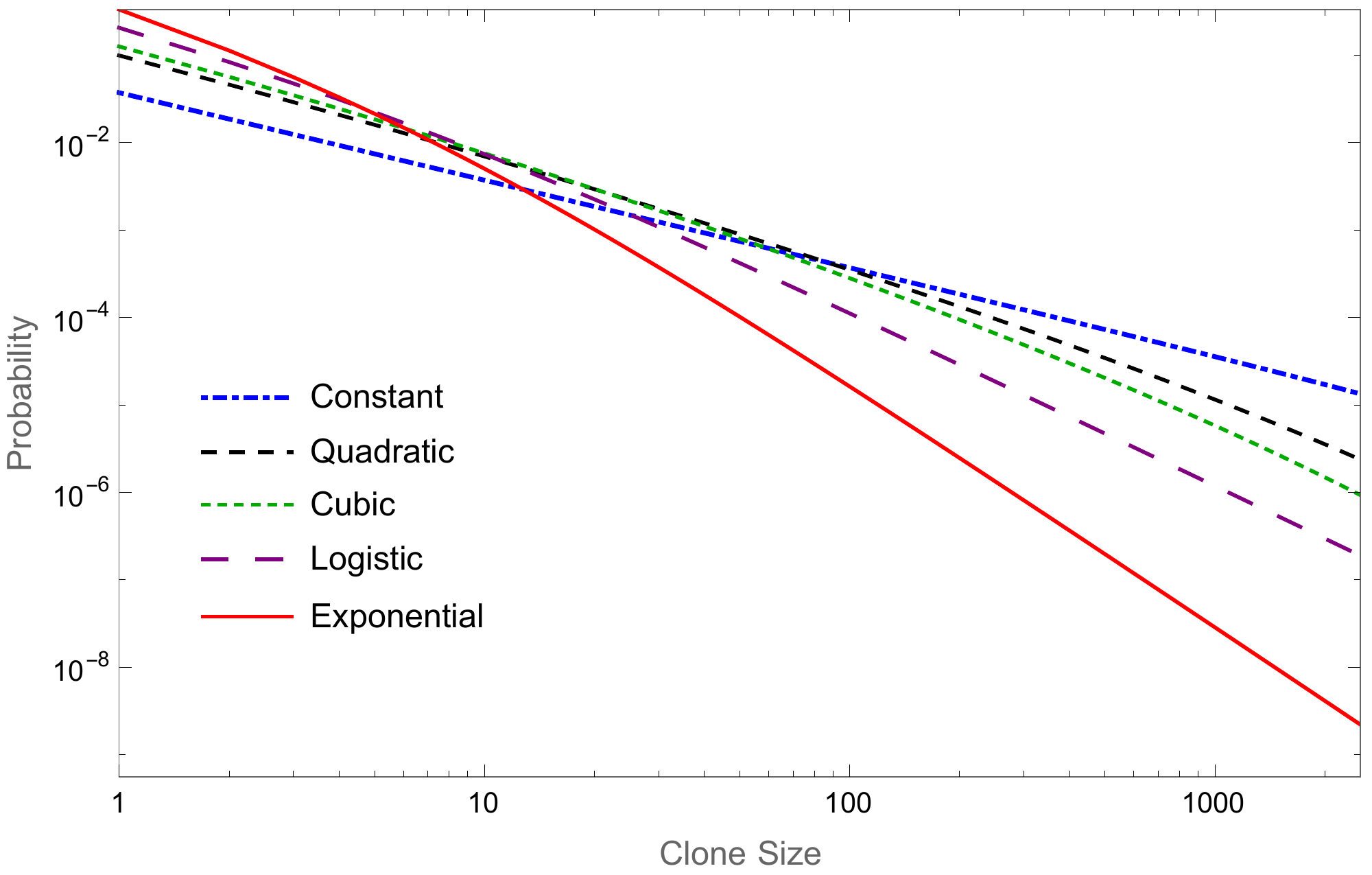}
\caption{}
\end{subfigure}
\caption{(a): Growth curves for different wild-type growth functions $n_{\tau}$. 
(b): The associated probability mass functions, derived in Section \ref{clonesize}, for the clone size when wild-type follows growth curves shown in (a). Parameters: $\delta=1.8, \,\lambda=1,\,t=9, K=20000$.}
\label{Dif}
\end{figure}

Agreeing with intuition for $K=1$ we recover the generating function of the constant case, and $\lim_{K\rightarrow \infty}\genclone(s)$ gives the generating function for exponential wild-type growth. Therefore the logistic case interpolates between the constant and exponential growth cases. The mass function can be obtained by expanding the non-logarithmic and logarithmic function in $\genclone(s)$ and using the Cauchy product formula. However, this method provides little insight and numerically it is simpler to use the fast Fourier transform.

\subsection{Monotone distribution and finite time cut-off}\label{secmonotone}
We conclude this section by demonstrating general features that exist in the clone size distribution at finite times. Again proofs are provided in Appendix \ref{clonalappend}. Firstly we see that, regardless of the particular wild-type growth function, the monotone decreasing nature of the mass function for the birth-death process is preserved in the clone size distribution.
\begin{prop}
   \label{pmfshape}
As long as $n_{\tau}$ is positive for some subinterval of $[0,t]$, then for $k\geq 1$ we have $
  \mathbb{P}(Y_{t}=k+1)< \mathbb{P}(Y_{t}=k)$ for any finite  $t>0$.
   \end{prop}
Whether $ \mathbb{P}(Y_{t}=0)\geq \mathbb{P}(Y_{t}=1)$ depends on $n_{\tau}$ and $t$, but the inequality is typically true for long times. 
Note that in contrast, the mass function of the total number of mutants is not monotone in general \cite{Keller:2015}.

Now restricting ourselves to the $\lambda>0$ case, as an example, consider the mass function when the size of the wild-type population is constant, which is given by \eqref{constpmf}, and specifically for $k\geq 1$. For any moderate $t$, $\mathcal{S}_{t}^{-1}$ is typically close to unity but for large $k$, $\mathcal{S}_{t}^{-k}$ will become the dominant term in the mass function, dictating exponential decay. We term this a cut-off in the distribution which occurs at approximately $k= O (e^{\lambda t})$. It is an artifact of the mass function for the birth-death process \eqref{BDpmf}. Hence we will have (at least) two behaviour regimes for the mass function for finite times. Here we show that this cut-off exists generally for finite times. 

\begin{theorem}\label{cutoff}
Let $\lambda>0$ and $n_{\tau}$ be continuous and positive for $\tau\in[0,t]$. Then
\begin{gather}
\mathbb{P}(Y_{t}=k)=\mathcal{S}_{t}^{-k}\Theta_t(k),
\end{gather}
where $\Theta_t(k)$ is an unspecified subexponential factor, i.e. $\limsup_{k\rightarrow \infty}\sqrt[k]{\Theta_t(k)}=1$, and  $\mathcal{S}_t$ is given by \eqref{ratio}.
\end{theorem}
Note that $\mathcal{S}_{t}>1$ for finite $t$, and $\mathcal{S}_{t}\rightarrow 1$ exponentially fast for large $t$. Hence the cut-off will disappear for long times and the subexponential factor, discussed in detail in Section \ref{longtime}, will completely determine the tail of the distribution. Also notice that the power-law cases, $n_{\tau}=\tau^{\rho}$, for $\rho\geq 1$ are not covered as, to make the analysis tractable, they artificially start at $n_{0}=0$. However the generating function in this case \eqref{powergf} also has its closest to origin singularity at $\mathcal{S}_{t}$ so the cut-off exists there also.

%%%%%%%%%%%%%%%%%%%%%%%%%%%%%%%%

%\tib{For example, exponential-type functions are $e^{\alpha x^\beta}$ for $\alpha\ne 0$, $\beta\ge 1$, $x^x$, while subexponential functions are $\log(x)$, $x^\alpha$, $e^{x^\beta}$ for $\beta<1$, $e^{x/\log(x)}$.}

\section{Universal large time features}\label{longtime}
Here we give results regarding the large time behaviour of our model which is relevant in many applications and also provides general insight. In many applications the cut-off location ($k=O(e^{\lambda t}$)) is so large that the distribution at or above this point is of little relevance and hence for this purpose the limiting approximations now discussed are of particular interest. Using the notation of Theorem \ref{cutoff}, this section investigates the large time form of $\Theta_{t}(k)$. The proofs for the results presented in this section can be found in Appendix \ref{unproofs}.
  We highlight the power-law decaying, ``fat'' tail found in each case. Henceforth we again assume $\lambda>0$, i.e. a supercritical birth-death process.
  
 \subsection{General wild type growth functions} 
To give general results we introduce the following assumption which will be assumed to hold for the remainder of this section.
\begin{assumption}\label{ass1}
For wild-type growth function $n_{\tau}$ we assume
\begin{enumerate}[(i)]
\item $n_{\tau}=0$ for $\tau<0$, continuous for $\tau> 0$ and right continuous at $\tau=0$.
\item $n_{\tau}$ is positive and monotone increasing for $\tau> 0$.
\item For $x\geq 0$ the limit
$\lim_{t\rightarrow\infty}n_{t-x}/n_{t}$
exists, is positive and finite.
 \end{enumerate}
 \end{assumption}
We note that the cases discussed in Section \ref{clonesize} are all covered by Assumption \ref{ass1}. The reason for the seemingly arbitrary limit assumed in (iii) becomes clear with the following result which is an application of the theory of regular variation found in \cite{Bingham:1987}.  
 \begin{lemma}\label{limform}

For $x\geq 0$
\begin{equation}
\lim_{t\rightarrow\infty}\frac{n_{t-x}}{n_{t}}=e^{-x\delta^*},\text{ where }\lim_{t\rightarrow\infty}\frac{\log n_{t}}{t}=\delta^*\geq 0.
\end{equation}
\end{lemma}
Often the long-time behaviour of the clone size distribution may be separated into $\delta^*>0$ and $\delta^*=0$, and so we give the following definition \cite{Flajolet:2009}.
\begin{definition}
Consider a real valued function $f(x)$ such that
\begin{equation}
\lim_{x\rightarrow\infty}\frac{\log f(x)}{x}=\delta^*
\end{equation}
holds for some $\delta^*\in\mathbb{R}$. Then $f(x)$ is of \textit{exponential-type} for $\delta^*\neq 0$ and is \textit{subexponential} for $\delta^*=0$. 
\end{definition}
Simple examples of subexponential functions are $e^{\sqrt{t}},\, \log(t)$, $t^{\rho}$, while $e^{\delta t}$, $e^{\delta t}t^{\rho}$ are of exponential-type, with $\delta,\rho\in \mathbb{R}$.

\subsection{Mean and variance}
We now address the asymptotic properties of the clone size distribution by first discussing its mean and variance. 

\begin{theorem}\label{meanlim}
With $s_{i}(t)$ subexponential functions such that $s_{1}(t),\,s_{3}(t)\rightarrow \infty$
 \begin{equation}
 \mathbb{E}(Y_{t})\sim \begin{cases}
  \frac{\delta^*}{\delta^*-\lambda} & \lambda<\delta^*\\
         s_{1}(t) & \delta^*=\lambda \\
      e^{(\lambda-\delta^*) t}s_{2}(t)&  \delta ^*<\lambda
  \end{cases}
\hspace{1.5cm}
  \mathrm{Var}(Y_{t})\sim
  \begin{cases}
  \frac{\delta^*}{\lambda}\left(\frac{2}{\delta^*-2\lambda}-\frac{2-\lambda}{\delta^*-\lambda}\right)-\left(\frac{\delta^*}{\delta^*-\lambda}\right)^2  & 2\lambda<\delta^*\\
         s_{3}(t) & \delta^*=2\lambda \\
    e^{(2\lambda-\delta^*)t}s_4(t) & \delta^*<2\lambda
  \end{cases}
  \end{equation}
  as $t\rightarrow \infty$.
\end{theorem}

%\begin{theorem}\label{meanlim}
%With $s_{i}(t)$ subexponential functions such that $s_{1}(t),\,s_{3}(t)\rightarrow \infty$
% \begin{equation}
% \mathbb{E}(Y_{t})\sim \begin{cases}
%  \frac{\delta^*}{\delta^*-\lambda} & \mbox{for }\lambda<\delta^*\\
%         s_{1}(t) & \mbox{for }\delta^*=\lambda \\
%      e^{(\lambda-\delta^*) t}s_{2}(t)& \mbox{for } \delta ^*<\lambda
%  \end{cases},
%\,
%  \mathrm{Var}(Y_{t})\sim
%  \begin{cases}
%  \frac{\delta^*}{\lambda}\left(\frac{2}{\delta^*-2\lambda}-\frac{2-\lambda}{\delta^*-\lambda}\right)-\left(\frac{\delta^*}{\delta^*-\lambda}\right)^2  & \mbox{for }2\lambda<\delta^*\\
%         s_{3}(t) & \mbox{for }\delta^*=2\lambda \\
%    e^{(2\lambda-\delta^*)t}s_4(t) & \mbox{for }\delta^*<2\lambda
%  \end{cases}
%  \end{equation}
%  as $t\rightarrow \infty$.
%\end{theorem}
The leading asymptotic behaviour which has different regimes dependent on $\delta^*/\lambda$ is illustrated in Figure \ref{illust}. As an example, for the exponential case $n_{\tau}=e^{\delta \tau}$, by using \eqref{meanlink}
   and the results found in \cite{Keller:2015}, then $\delta^*=\delta,\,s_{1}(t)=\lambda t,\,s_{2}(t)=\frac{\delta}{\lambda-\delta},\,s_{3}(t)=4 t $ and $s_{4}(t)=\frac{2\delta}{\lambda(2\lambda-\delta)}$. 
\begin{figure}[htbp]
\begin{center}
\setlength{\unitlength}{2.6cm}
\begin{picture}(6,0.6)(-3,-0.3)
\put(-1.7,0){\vector(1,0){3.3}}
\put(-1.7,-.05){\line(0,1){0.1}}
\put(-1.73,-0.2) {$0$}
\put(-0.7,-.05){\line(0,1){0.1}}
\put(-.74,-0.2) {$1$}
\put(.3,-.05){\line(0,1){0.1}}
\put(0.34,-0.2) {$2$}
\put(-2.54,0.3) {$\mathbb{E}(Y_t)$}
\put(-1.35,0.3) {expo}
\put(-.35,0.3) {const}
\put(.65,0.3) {const}
\put(-2.54,0.1) {$\mathrm{Var}(Y_t)$}
\put(-1.35,0.1) {expo}
\put(-.35,0.1) {expo}
\put(.65,0.1) {const}
\put(-1.35,-0.2) {fit}
\put(-1.35,-0.35) {mutants}
\put(-.35,-0.2) {unfit}
\put(-.35,-0.35) {mutants}
\put(.65,-0.2) {more unfit}
\put(.65,-0.35) {mutants}
\put(1.75,-0.) {$\gamma^*=\delta^*/\lambda$}
\end{picture}
\end{center}
\caption{Illustration of the asymptotic behaviour of the mean and variance as given in Theorem \ref{meanlim}.}
\label{illust}
\end{figure}
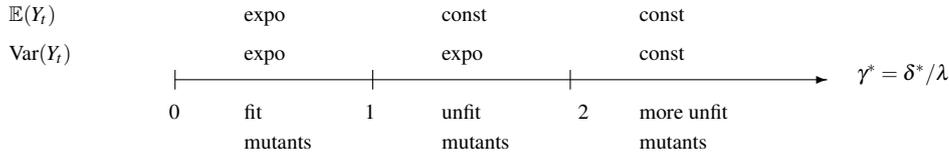 

\subsection{Large time clone size distribution}
Turning to the distribution function we have the following result regarding the generating function at large times.
\begin{theorem}\label{bigthm}
Let $\gamma^*= \delta^*/\lambda$. Then for $|s|<1$
\begin{equation}\label{limgf}
\lim_{t\rightarrow \infty}\frac{a_{t}}{n_{t}}(\genclone(s)-\beta)=\frac{1}{\gamma^*}\left[1-\Fze{1,\gamma^*}{1+\gamma*}{\xi}\right]=-\sum_{k\ge 1}\frac{\xi^{k}}{\gamma^*+k}.
\end{equation}
\end{theorem}

This result is made clearer in the next corollary,  in which the cases of exponential-type and subexponential growth are separated. This is as, for $\delta^*>0$,
 \begin{equation}
\lim_{t\rightarrow\infty}\frac{n_{t}}{a_{t}}\rightarrow \delta^* .
\end{equation}
For a proof see Lemma \ref{subexp1}. Consequently in the exponential-type setting, the limiting result is a proper probability distribution, whilst in the subexponential case it is not. We can interpret this as the clone sizes staying finite in the exponential case but grow to infinity for subexponential cases at large times. Henceforth, for brevity we do not impose such a separation but the reader should note that for exponential-type growth the above limit holds and may simplify further results.
\begin{corollary}\label{bigcor}
For $|s|<1$,
\begin{align}
&\lim_{t\rightarrow \infty}(\genclone(s)-\beta)=\lambda\left[1-\Fze{1,\gamma^*}{1+\gamma*}{\xi}\right]  & \gamma^*>0,\\
&\lim_{t\rightarrow \infty}\frac{a_{t}}{n_{t}}(\genclone(s)-\beta)=\log(1-\xi) \,\hspace{17mm} &  \gamma^*=0,
\end{align}
where the second expression is the $\gamma^*\to 0$ limit of the first expression. 
Then for $t\rightarrow \infty$ the probabilities for exponential-type growth $\gamma^*>0$ are
\begin{equation}
\mathbb{P}(Y_{t}=k)\sim 
      \begin{cases}
     1-\lambda \Fze{1,\gamma^*}{1+\gamma^*}{\beta} &   k=0\\      
     \frac{\delta^*\Gamma(k)}{(\gamma^*+1)_{k}}\Fze{k,\gamma^*}{1+\gamma^*+k}{\beta} &   k\geq 1,
     \end{cases}
\end{equation}
and for subexponential growth ($\gamma^*=0$)
\begin{equation}
\label{powerlimit}
\mathbb{P}(Y_{t}=k)\sim 
      \begin{cases}
     \beta+\frac{n_{t}\log(\lambda)}{a_{t}\alpha} &  k=0\\
      \frac{n_{t}}{a_{t} \alpha k} &   k\geq 1.
     \end{cases} 
\end{equation} 
\end{corollary}
The expressions obtained in the $\delta^*>0$ case also appeared as an approximation in \cite{Kessler:2015} for the total number of mutants with stochastic wild-type and mutant growth when the mean number of clones is small. This can now be interpreted as an application of Proposition \ref{smallarrive}.

\begin{figure}[t!]
\centering
\includegraphics[width=.7\linewidth,height=6cm]{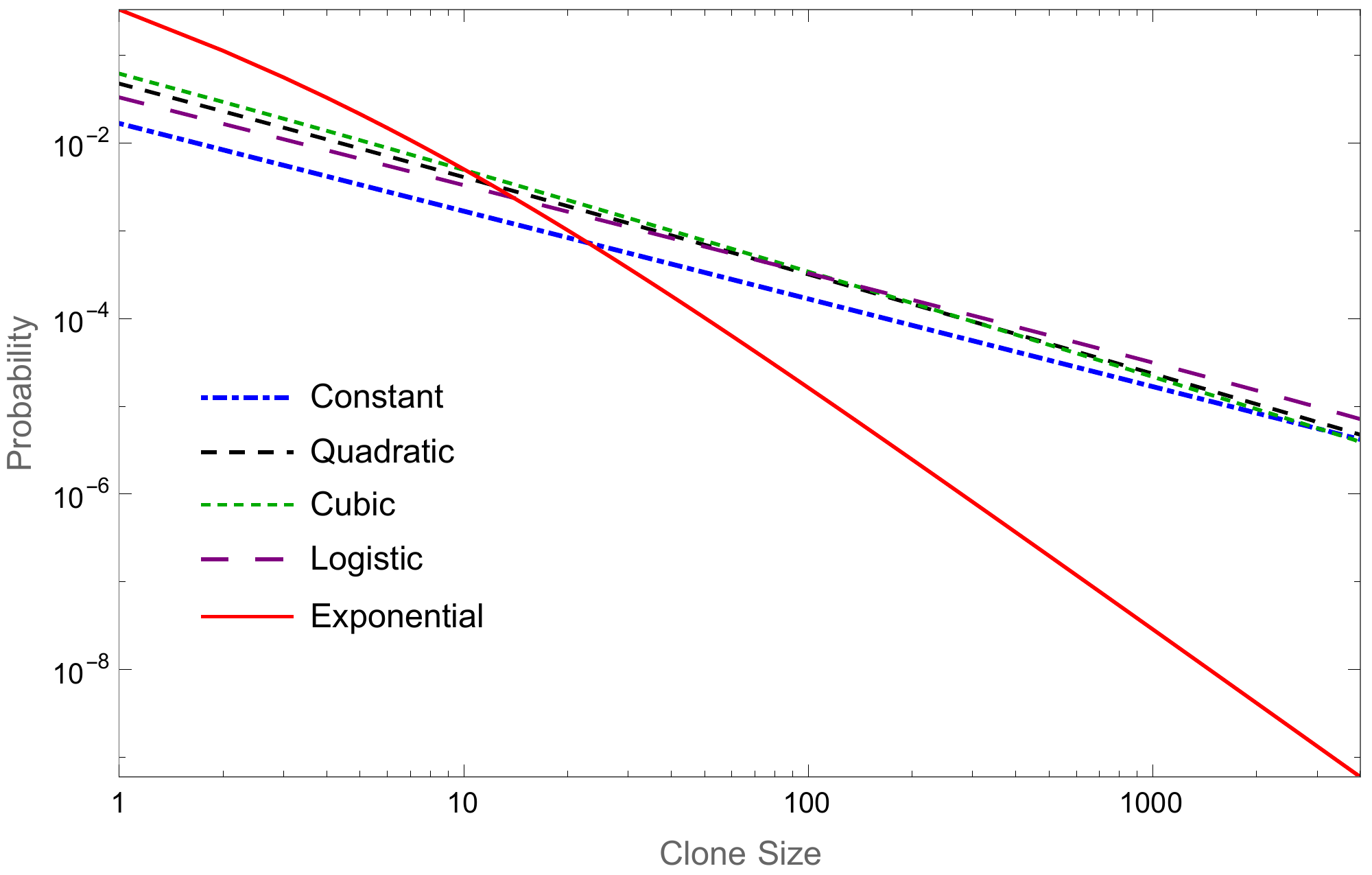}
\caption{Transition to the asymptotic regime as described in Corollary \ref{bigcor}. For subexponential wild-type growth the mass functions tend to $k^{-1}$ behaviour, while for exponential-type it tends to $k^{-1-\gamma^*}$. Here $t=20$ and all other parameters are as given in Figure \ref{Dif}.}
\end{figure}

The case of immortal mutants does not simplify the above expressions for subexponential growth, but for exponential-type growth, by applying \eqref{hypgeom3} then \eqref{hypgeom1} to the limiting generating function, we have the following link to the Yule-Simon distribution which appears often in random networks \cite{Simon:1955,Krapivsky:2001}.
\begin{corollary}
For immortal mutants with exponential-type wild-type growth the clone size distribution $Y_{t}$ follows a Yule-Simon distribution with parameter $\delta^*$ for large times. That is, for $\beta=0,\,\delta^*>0$,
\begin{equation}
\lim_{t\rightarrow \infty}\genclone(s)=\frac{s\delta^*}{\delta^*+1}\Fze{1,1}{2+\delta^*}{s},
\end{equation}
and thus, for $k\geq 1$,
\begin{equation}
\lim_{t\rightarrow\infty}\mathbb{P}(Y_{t}=k)=\frac{\delta^*\Gamma(k)}{(\delta^*+1)_{k}} .
\end{equation}
\end{corollary}
With immortal, neutral ($\delta^*=1$)  mutants we have
\begin{equation}
 \lim_{t\rightarrow\infty} \pr(Y_t=k) = \frac{1}{ \alpha k(k+1)}.
\end{equation}
which is in agreement with the long time limit of \eqref{neutimmtime}.
For immortal mutants and exponential-type growth, as the clone size distribution tends to a Yule-Simon distribution, we expect power-law tail behaviour at large times \cite{Newman:2005}. Interestingly, we see that this behaviour holds when we include mutant death and have general wild-type growth.
\begin{corollary}\label{tailcor}
At large times, the tail of the clone size distribution follows a power-law with index $1+\gamma^*$. More precisely,
\begin{equation}
\label{tailexp}
\lim_{k\rightarrow\infty}\lim_{t\rightarrow\infty}\frac{k^{\gamma^*+1}a_{t}}{n_{t}}\mathbb{P}(Y_{t}=k)=\frac{\Gamma(1+\gamma^*) }{\lambda^{\gamma^*}}.
\end{equation}
\end{corollary}

\subsection{Large time distribution for total number of mutants}

Finally, to conclude this section we give the corresponding results for the total number of mutants $B_{t}$ in the often used \textit{Large Population-Small Mutation} limit.

\begin{theorem}\label{largepopsmallmut}
Letting $\theta=\mu n_{t}$ be constant and with $s_{i}(t)$ subexponential functions as in Theorem \ref{meanlim}
 \begin{equation}
 \mathbb{E}(B_{t})\sim \begin{cases}
 \frac{\theta a_{t}}{n_{t}} \frac{\delta^*}{\delta^*-\lambda} & \lambda<\delta^*\\
        \frac{\theta a_{t}}{n_{t}}  s_{1}(t) & \delta^*=\lambda \\
   \frac{\theta a_{t}}{n_{t}}e^{(\lambda-\delta^*) t}s_{2}(t)&  \delta ^*<\lambda
  \end{cases}
\hspace{1.5cm}
  \mathrm{Var}(B_{t})\sim
  \begin{cases}
 \frac{\theta a_{t}}{n_{t}} \frac{\delta^*}{\lambda}\left(\frac{2}{\delta^*-2\lambda}-\frac{2-\lambda}{\delta^*-\lambda}\right)  & 2\lambda<\delta^*\\
      \frac{\theta a_{t}}{n_{t}} s_{3}(t) & \delta^*=2\lambda \\
       \frac{\theta a_{t}}{n_{t}}e^{(2\lambda-\delta^*)t}s_4(t) & \delta^*<2\lambda
  \end{cases}
  \end{equation}
  as $t\rightarrow \infty$. For $|s|<1$
\begin{equation}
\lim_{\substack{t\rightarrow \infty\\ \theta \,\mathrm{  constant}}}\gentotal(s)\exp\left(\frac{\theta a_{t}\lambda}{n_{t}}\right)=\exp\left(\frac{\theta}{\gamma^*}\left[1-\Fze{1,\gamma^*}{1+\gamma*}{\xi}\right]\right),
\end{equation}
and we have the following tail result
\begin{equation}
\lim_{k\rightarrow\infty}\lim_{\substack{t\rightarrow \infty\\ \theta \,\mathrm{  constant}}}k^{\gamma^*+1}\exp\left(\frac{\theta a_{t}}{n_{t}}\right)\mathbb{P}(B_{t}=k)=\frac{\theta\Gamma(1+\gamma^*) }{\lambda^{\gamma^*}}.
\end{equation}
\end{theorem}

%\begin{theorem}\label{largepopsmallmut}
%Letting $\theta=\mu n_{t}$ be constant and with $s_{i}(t)$ subexponential functions as in Theorem \ref{meanlim}
% \begin{equation}
% \mathbb{E}(B_{t})\sim \begin{cases}
% \frac{\theta a_{t}}{n_{t}} \frac{\delta^*}{\delta^*-\lambda} & \mbox{for }\lambda<\delta^*\\
%        \frac{\theta a_{t}}{n_{t}}  s_{1}(t) & \mbox{for }\delta^*=\lambda \\
%   \frac{\theta a_{t}}{n_{t}}e^{(\lambda-\delta^*) t}s_{2}(t)& \mbox{for } \delta ^*<\lambda
%  \end{cases},
%\,
%  \mathrm{Var}(B_{t})\sim
%  \begin{cases}
% \frac{\theta a_{t}}{n_{t}} \frac{\delta^*}{\lambda}\left(\frac{2}{\delta^*-2\lambda}-\frac{2-\lambda}{\delta^*-\lambda}\right)  & \mbox{for }2\lambda<\delta^*\\
%      \frac{\theta a_{t}}{n_{t}} s_{3}(t) & \mbox{for }\delta^*=2\lambda \\
%       \frac{\theta a_{t}}{n_{t}}e^{(2\lambda-\delta^*)t}s_4(t) & \mbox{for }\delta^*<2\lambda
%  \end{cases}
%  \end{equation}
%  as $t\rightarrow \infty$. For $|s|<1$
%\begin{equation}
%\lim_{\substack{t\rightarrow \infty\\ \theta \,\mathrm{  constant}}}\gentotal(s)\exp\left(\frac{\theta a_{t}\lambda}{n_{t}}\right)=\exp\left(\frac{\theta}{\gamma^*}\left[1-\Fze{1,\gamma^*}{1+\gamma*}{\xi}\right]\right),
%\end{equation}
%and we have the following tail result
%\begin{equation}
%\lim_{k\rightarrow\infty}\lim_{\substack{t\rightarrow \infty\\ \theta \,\mathrm{  constant}}}k^{\gamma^*+1}\exp\left(\frac{\theta a_{t}}{n_{t}}\right)\mathbb{P}(B_{t}=k)=\frac{\theta\Gamma(1+\gamma^*) }{\lambda^{\gamma^*}}.
%\end{equation}
%\end{theorem}

%%%%%%%%%%%%%%%%%%%%%%%%%%%%%%%%  
\section{Tail behaviour in empirical metastatic data} \label{applications}

Given the above discussion we expect, for a large class of wild-type growth functions, to see power tail behaviour on approach to the exponential cut-off in the clone size distribution. We take the first steps to verify this theoretical hypothesis by analysing an empirical metastatic data. In this setting the wild-type population is the primary tumour and mutant clones are the metastases.
 
Our data is sourced from the supplementary materials in \cite{Bozic:2013}. This data is taken from 22 patients;  7 with pancreatic ductal adenocarcinomas, 11 with colorectal carcinomas, and 6 with melanomas. One patient had only a single metastasis so we discard this data. Of the 21 remaining patients the number of cells in a single metastasis ranged from $6\times 10^6$ to $2.23\times 10^9$. Our theoretical model predicts a cut-off in the distribution around $k=e^{\lambda t}$. Taking some sample parameters from the literature, namely $\lambda=0.069$/day \cite{Diaz:2012}, and $t=14.1$ years \cite{Yachida:2010}, this leads to a cut-off around $k\approx10^{154}$ cells. Due to the enormity of this value we ignore the cut-off here. Additionally, as the minimum observed metastasis size is $6\times10^6$ cells, we assume that all data points are sampled from the tail of the distribution.

%  Additionally, as the size distribution is monotonically decreasing past $k=1$ (Proposition \ref{pmfshape}) and the minimum metastasis size of the patients considered is $6\times10^6$ cells, we assume that all data points are sampled from the tail of the distribution.

For each of the data-sets we examine the likelihood ratio to determine if the data is more likely sampled from a power-law decaying or geometrically decaying distribution. 19 of the 21 data-set return the power-law hypothesis as more plausible which is in agreement with the theoretical prediction. Both are single parameter distributions and maximum likelihood analysis was utilised to estimate the parameters. The methodology outlined in \cite{Clauset:2009} was broadly followed and brief details regarding calculating maximum likelihood estimates (MLEs) are given in Appendix \ref{MLEapp}. We note that in this context the likelihood ratio point esimator returns equivalent results to the Akaike information criterion widely used in model selection \cite{Burnham:1998}. Under the power-law model, $\mathbb{P}(Y_{t}=k)\propto k^{-\omega}$, for 20 of the 21 data-sets we find the point estimate of the power-law index, $\hat \omega$, lies in $[-2,-1]$. The outlier comes from the smallest data-set (3 metastases). Due to the small size of data-sets, we recognise the influence of statistical fluctuations. 
\begin{figure}[t!]
   \centering
     \begin{subfigure}[b]{0.32\textwidth}
              \includegraphics[width=\textwidth,height=5cm]{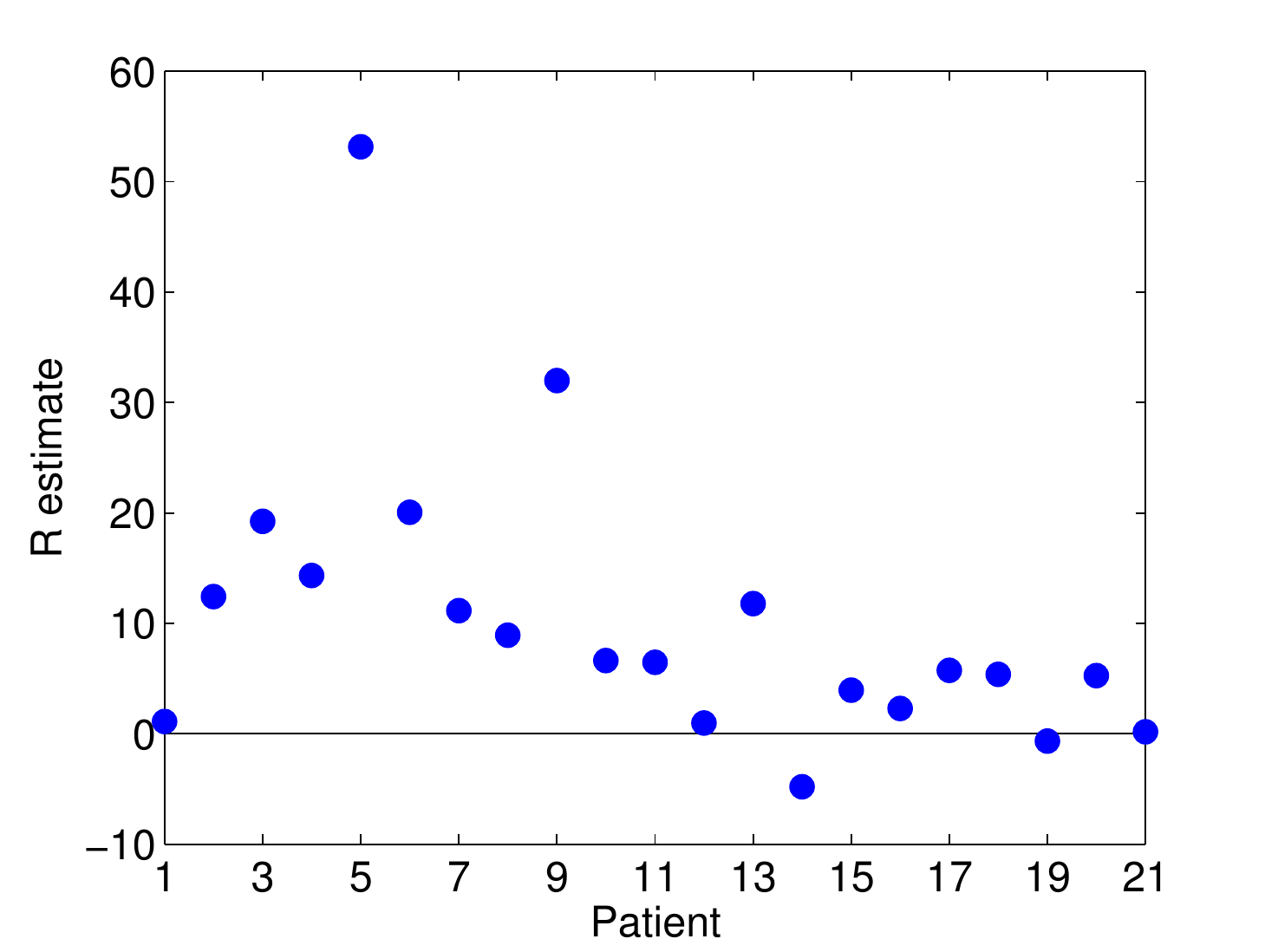}
        \caption{}
        \label{Rest}
    \end{subfigure}
    \begin{subfigure}[b]{0.32\textwidth}
           \includegraphics[width=\textwidth,height=5cm]{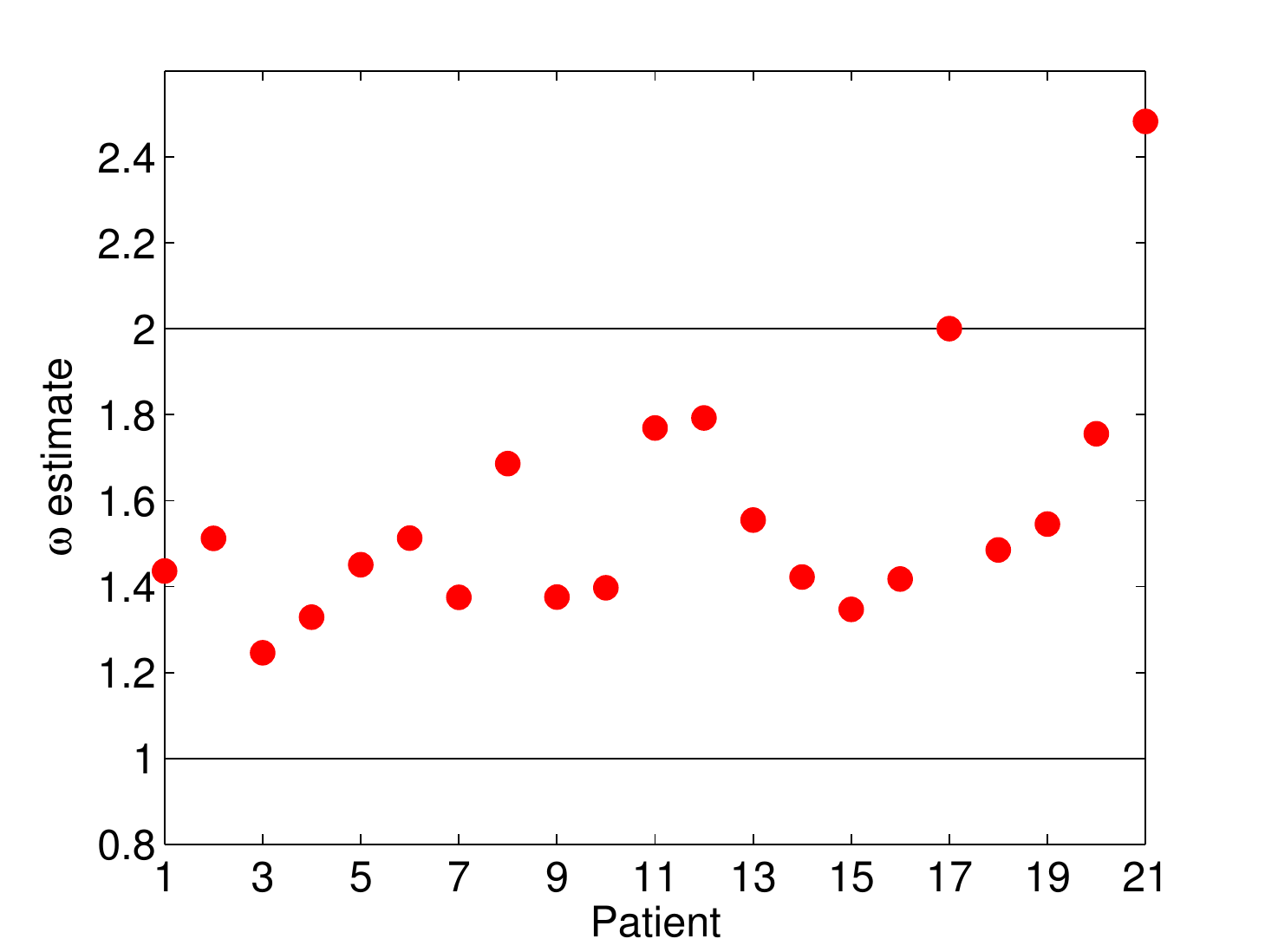}
        \caption{}
        \label{alphaest}
    \end{subfigure}
    ~ %add desired spacing between images, e. g. ~, \quad, \qquad, \hfill etc. 
      \begin{subfigure}[b]{0.32\textwidth}
         \includegraphics[width=\textwidth,height=5cm]{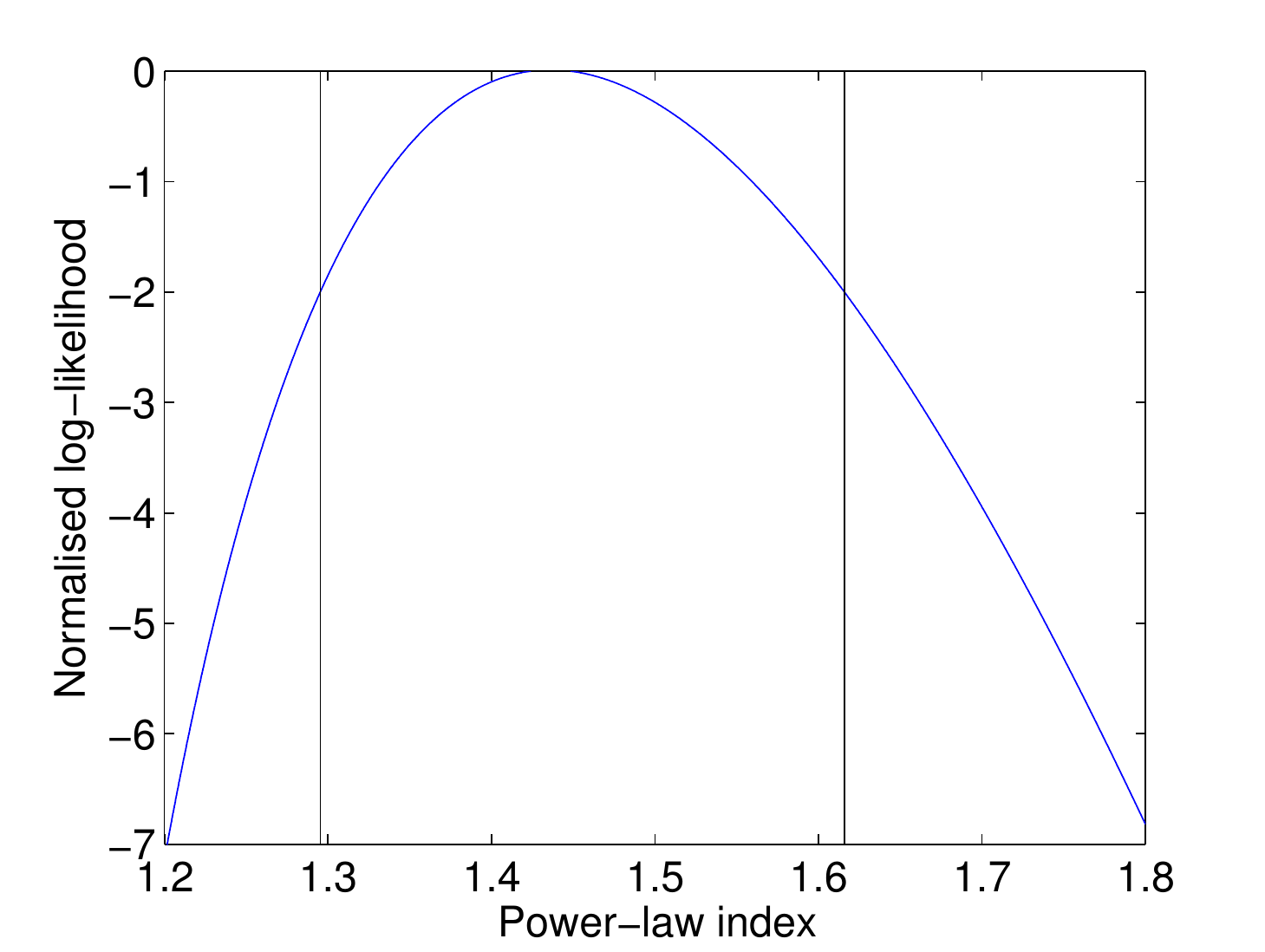}
        \caption{}
        \label{Loglike}
    \end{subfigure}      
      ~ %add desired spacing between images, e. g. ~, \quad, \qquad, \hfill etc. 
      \caption{Likelihood analysis results: Patients are sorted left to right by number of metastases with patient 1 having 30 mets to patient 21 having 3. Hence values to left of figures are more significant. (A) Likelihood ratio $\hat{\mathcal{R}}$ for each data-set. Points above the horizontal line suggest data-set is from a power-law distribution over a geometric distribution. (B) Estimate $\hat \omega$ for each data-set, determined via maximum likelihood. (C) Normalised log-likelihood function for first best set. Vertical bars show the likelihood interval.}
\end{figure}

Details of the likelihood ratio are as follows. Let $\textbf{y}=(y_{1},\ldots,y_{N})$ be a data-set of size $N$. We test the hypothesis that $\textbf{y}$ is drawn from a power-law distribution, $\mathbb{P}_{1}(Y_{t}=k)=C_{1}k^{-\omega}$, versus that it is sampled from a geometric distribution, $\mathbb{P}_{2}(Y_{t}=k)=C_{2}p(1-p)^{k}$, where $C_{1},\,C_{2}$ are normalising constants and $p$ is the parameter for the geometric distribution. The log-likelihood ratio is
\begin{gather}
\hat{\mathcal{R}}=\sum_{i=1}^{N}[\log\mathbb{P}_{1}(Y_{t}=y_{i})-\log\mathbb{P}_{2}(Y_{t}=y_{i})],
\end{gather}
where $\hat{\mathcal{R}}>0$ gives support to the hypothesis that the data is drawn from the  power-law distribution with MLE exponent $\hat \omega$, over the geometric distribution with MLE parameter $\hat p$. The results are given in Figure \ref{Rest}.

Assuming a power-law distribution the maximum likelihood estimates for the exponent $\omega$ for each data-set are given in Figure \ref{alphaest}. Due to the small sample size of our data-sets and the high variance in the distribution, we do not derive confidence intervals via normal distribution approximations. Instead we show the normalised log-likelihood, $\log\LTL(\omega)/\LTL(\hat \omega)$, for our best data-set, with $N=30$, in Figure \ref{Loglike}, where $\LTL(\omega)$ is the likelihood function. Also, following \cite{Hudson:1971}, we demonstrate the likelihood interval defined as
\begin{gather}
I(\omega)=\{\omega:\log \frac{\LTL(\omega)}{\LTL(\hat \omega)}\geq -2\}.
\end{gather}
If a large sample size was possible this interval would correspond to a $95.4\%$ confidence
interval.
For the data-set with $N=30$ we numerically determined $I(\omega)=[1.295,1.616]$, demonstrated as the domain between the vertical bars in Figure \ref{Loglike}.

%%%%%%%%%%%%%%%%%%%%%%%%%%%%%%%%  
\section{Alternative approaches}\label{alt}

\subsection{Deterministic approximation}
In order to circumvent the complexity introduced by the birth-death process one might be tempted to simply assume the mutant clone size grows according to $e^{\lambda \tau}$, the mean of the birth-death process. This approach corroborates our results regarding the tail of the size distribution. Indeed, the clone size density may be found to be
\begin{gather}\label{dens}
f_{Y_t}(y)=\frac{n_{t-\frac{\log(y)}{\lambda}}}{a_{t}\lambda y}.
 \end{gather}
which has support $[1,e^{\lambda t}]$. This formula can also be found in \cite{Hanin:2006}. Then, as in Section \ref{longtime} under Assumption \ref{ass1},
\begin{equation}
\lim_{t\rightarrow \infty}\frac{a_{t}}{n_{t}}f_{Y_t}(y)=\frac{1}{y^{\gamma^*+1}\lambda}.
\end{equation}
Thus, asymptotically the density has the same behaviour as the tail of the limiting result given in Corollary \ref{tailcor}, but with a different amplitude.

%Then letting $x=\log(y)$, we obtain
%\begin{gather}
%\label{logdens}
%\partial_{x}\log(f_{Y_t}(x))=-\bigg(\frac{n'_{t-x/\lambda}}{\lambda n_{t-x/\lambda}}+1\bigg),
%\end{gather}
%defined for $x\in[0,\lambda t]$. Here $n'_{u}=\frac{d}{d u}n_{u}$. Thus we see that if $\lim_{t\rightarrow \infty}\log(n_{t})'=0$, then $f_{Y_{t}}(y)\propto y^{-1}$ for long times.

%Thus we see that if the long-time wild-type growth function is subexponential, which here implies $\lim_{t\rightarrow \infty}\log(n_{t})'=0$, then $f_{Y_{t}}(y)\propto y^{-1}$ for long times.

 However despite this agreement the densities given by \eqref{dens} for specific wild-type growth function differ significantly compared with stochastic mutant proliferation.  Letting $Y_{t}^{\mathrm{Stoch}}$ be the clone size distribution with stochastic mutant growth and $Y_{t}^{\mathrm{Det}}$ be its deterministic approximation specified by \eqref{dens}, we may quantify the approximation error, at least for the moments, by the following theorem, whose proof can be found in Appendix \ref{altproofs}.
\begin{theorem}
\label{thmcums}
As $t\rightarrow \infty$
\begin{gather}
\frac{\mathbb{E}[(Y_{t}^{\mathrm{Stoch}})^{m}]}{\mathbb{E}[(Y_{t}^{\mathrm{Det}})^{m}]}=\frac{m!}{\lambda^{m-1}}+O(e^{-\lambda t}).
\end{gather}
\end{theorem}

\subsection{Time-dependent rate parameters}
Some authors \cite{Houchmandzadeh:2015,Tomasetti:2012} have previously considered the case where all rates in the system are multiplied by a time-dependent function, say $z(\tau)$. This is relevant in the scenario where both the wild-type and mutant populations have their growth restricted simultaneously by environmental factors, for example exposure to a chemotherapeutic agent. We observe that under a change of timescale this system is equivalent to our setting with exponential wild-type growth. This is due to the following argument.

 In this setting the wild-type population is governed by
\begin{gather}
\label{TDode}
\frac{d n_{\tau}}{d \tau}=\lambda z(\tau)n_{\tau}.
\end{gather}
Mutant clones are now initiated at a rate $\mu z(\tau) n_{\tau}$.  Let $\widehat{Z_{t}}$ be the size of a mutant population governed by the birth-death process with time-dependent rates. Once initiated, the size distribution obeys the forward Kolmogorov equation for time-dependent stochastic mutant proliferation 
\begin{equation} \label{TDFK}
\begin{split}
\partial_{t}  \mathbb{P}(\widehat{Z_t}=k) & =  \alpha z(t)(k-1)\mathbb{P}(\widehat{Z_t}=k-1)\\
 & + \beta z(t)(k+1)  \mathbb{P}(\widehat{Z_t}=k+1)-(1+\beta) z(t)k \mathbb{P}(\widehat{Z_t}=k).
\end{split}
\end{equation}
%\begin{align}
%\partial_{t}  \mathbb{P}(\widehat{Z_t}=k)= \alpha z(t)(k-1)\mathbb{P}(\widehat{Z_t}=k-1)
%\notag
%\\
%&+\beta z(t)(k+1)  \mathbb{P}(\widehat{Z_t}=k+1)-(1+\beta) z(t)k \mathbb{P}(\widehat{Z_t}=k).
%\label{TDFK}
%%(\alpha+\beta) z(t)k \mathbb{P}(\widehat{Z_t}=k).
%\end{align}
%Note the we still scale time to  have unit birth rate. This is the system we find in \cite{Houchmandzadeh:2015} (section B) and \cite{Tomasetti:2012} (Theorem 1). The conversion
% \begin{equation}
%\partial_{t}\mathbb{P}(\widehat{Z_t}=k)=\partial_n \mathbb{P}(\widehat{Z_t}=k)\frac{d n_{t}}{dt}
%\end{equation}
%can be used to make time the independent variable in the model of \cite{Houchmandzadeh:2015}. 
If we let
\begin{equation}
F(\tau)=\int_{0}^{\tau}z(s)ds
\end{equation}
then under a new timescale, $\tau'=F^{-1}(\tau)$ , the mutant clone initiation will occur at a rate $\mu n_{\tau'}$. Further, using the chain rule to express \eqref{TDode} and \eqref{TDFK} in terms of $\tau'$ we see that $n_{\tau'}=e^{\lambda \tau'}$ and that the forward Kolmogorov equation \eqref{TDFK} becomes \eqref{FK}. Thus, under a time-rescaling, all dynamics are equivalent to the system with exponential wild-type growth and stochastic mutant proliferation with constant birth and death rates, as studied in this article or in \cite{Keller:2015}.

\subsection{Poisson process characterisation of tail}

Complementing Corollary \ref{tailcor} in Section \ref{longtime}, following \cite{Tavare:1987}, we can also describe the size distribution for large clones at long times via a Poisson process in the following way. Let $(Z^{(i)}(t))_{i\geq1}$ be independent copies of the birth-death process as in Section \ref{theory} and $(T_{i})_{i\geq1}\subset(0,\infty)$ be the points of a of Poisson process with intensity $\mu n_{\tau}$, for $\tau\geq 0$. The $T_{i}$ represent the clone arrival times and so $K_t$ is the number of $T_{i}$ less than or equal to $t$.

Let us consider the size of the first clone. By known results about the large time behaviour of the birth-death process \cite{Athreya:2004}, as $t\rightarrow\infty$,
\begin{equation}
e^{-\lambda t}Z^{(1)}(t-T_{1})=e^{-\lambda T_{1}}e^{-\lambda(t-T_{1})}Z^{(1)}(t-T_{1})\rightarrow e^{-\lambda T_{1}}W_{1} \mbox{ a.s.   }
\end{equation}
The distribution of the limiting random variable $W_{1}$ is composed of a point mass at $0$ and an exponential random variable, precisely
\begin{equation}
\pr(W_{1}\leq x)=\beta+\lambda (1- e^{-\lambda x}),\quad x\geq 0.
 \end{equation} 
Analogously, with the details given in \cite{Tavare:1987} (Theorem 3), the limiting behaviour of the time-ordered clone sizes is given by
\begin{equation}
\lim_{t\rightarrow\infty} e^{-\lambda t}(Z^{(i)}(t-T_{i}))_{i\geq1}=(e^{-\lambda T_{i}}W_{i})_{i\geq1}\mbox{ a.s.}
\end{equation}
where $W_{1}$ is as before and all $W_{i}$ are \textit{iid}. The random sequence $(e^{-\lambda T_{i}}W_{i})_{i\geq1}$ takes non-negative real values, however if we restrict our attention to only the positive elements (that is clones that do not die), then these can be taken to be points from a non-homogeneous Poisson process. More precisely, the set $\{\sigma_{j}\}_{j\geq 1}\coloneqq\{e^{-\lambda T_{i}}W_{i}\}_{i\geq1}\setminus \{0\}$ are the points (in some order) from a Poisson process on $(0,\infty)$ with mean measure
\begin{equation}\label{genmean}
m(x,\infty)=\mu\int_{x}^{\infty}n_{\lambda^{-1}\log(s/x)}\frac{e^{-\lambda s}}{s}\,ds,\quad x>0.
\end{equation}
The proof of the above only requires minor modification from that of Theorem 4 in \cite{Tavare:1987}.

The Poisson process description of the large clones, at large times can also offer insight into further properties of the system, including links to the Poisson-Dirichlet distribution, see \cite{Tavare:1987,Durrett:2015}. With regards to the present article, the interesting point is that for fixed $\epsilon>0$, as the number of $\sigma_{j}>\epsilon$ is finite almost surely, we may sample unformly from this set (i.e. $\{\sigma_{j}\}_{j\geq1}\cap(\epsilon,\infty)$) and construct a random variable $Y_{\epsilon}$ with distribution
\begin{equation}
\pr(Y_{\epsilon}>x)=\frac{m(x,\infty)}{m(\epsilon,\infty)},\quad x\geq \epsilon
\end{equation}
where $m(x,\infty)$ is as in \eqref{genmean}. The new variable $Y_{\epsilon}$ can be related to the previously considered random variable $Y_t$ by the following result, whose proof is contained in Appendix \ref{altproofs}.
\begin{theorem}\label{Poislink}
For $\epsilon>0$, with $Y_{\epsilon}$ as above,
\begin{equation}
\lim_{t\rightarrow\infty}\pr(Y_{t}e^{-\lambda t}>x|Y_{t}e^{-\lambda t}>\epsilon)=\pr(Y_{\epsilon}>x),\quad x\geq \epsilon.
\end{equation}
\end{theorem}
Of note is the reappearance of power-law behaviour with a cut-off in the density of $Y_{\epsilon}$. For example in the constant wild-type case, $n_{\tau}=1$, the density, using \eqref{genmean}, is given by
\begin{equation}
f_{Y_{\epsilon}}(x)=\frac{d}{dx}\pr(Y_{\epsilon}\leq x)=\frac{e^{-\lambda x}}{x\Gamma(0,\lambda \epsilon)},\quad x\geq \epsilon.
\end{equation}
For exponential growth with neutral mutants, $n_{\tau}=e^{\lambda \tau}$,
\begin{equation}
f_{Y_{\epsilon}}(x)=\frac{e^{-\lambda x}}{x^2 }(1+\lambda x)\epsilon e^{\lambda \epsilon},\quad x\geq \epsilon.
\end{equation}
Note that the exponents in the power-law terms is equal to that given in Corollary \ref{tailcor}, indicating the two approaches are complimentary.

%%%%%%%%%%%%%%%%%%%%%%%%%%%%%%%%  
\section{Discussion}\label{discussion}

 In this study we focus on the size distribution for mutant clones initiated at a rate proportional to the size of the wild-type population. The size of the wild-type population is dictated by a generic deterministic growth function and the mutant growth is stochastic. This shifts the focus from previous studies which have mostly been concerned with exponential, or mean exponential, wild-type growth, and considered the total number of mutants. Results for the total number of mutants are, however, easily obtained from the clone size distribution.

The special cases of exponential, power-law and logistic wild-type growth were treated in detail, due to their extensive use in models for various applications. Utilising a generating function centred approach, exact time-dependent formulas were ascertained for the probability distributions in each case. Regardless of the growth function, the mass function is monotone decreasing and the distribution has a cut-off for any finite time. 
This cut-off goes to infinity for large times and is often enormous in practical applications, hence we focused on the approach to the cut-off.

We found that the clone size distribution behaves quite distinctly for exponential-type versus subexponential wild-type growth. Although the probability of finding a clone of any given size stays finite as $t\to\infty$ for exponential-type growth, it tends to zero for subexponential type.
Despite these differences, with a proper scaling, 
for a large class of growth functions we proved that the clone size distribution has a universal long-time form. This long-time form possesses a power-law, ``fat'' tail which decays as $1/k$ for subexponential wild-type growth, but faster for exponential-type growth. This can be intuitively understood as the tail distribution represents clones that arrive early, and the chance that a clone is initiated early in the process is larger for a slower growing wild-type function.
Hence we expect a ``fatter'' tail in the subexponential case.

%\mn{
%For exponential-type growth our limiting result is the corresponding clone size distribution for the \textit{Large Population-Small Mutation} limit given in \cite{Keller:2015} for the total mutant number distribution. That result, whilst also using an exponential determinstic wild-type growth function, coincided with the analogous limit when both wild-type and mutants proliferate stochastically, a two-type branching process, given in \cite{Kessler:2015}. Hence one might expect that the limiting clone size distributions given here also hold for the two-type branching process setting, although this has not been demonstrated.}
Note that although we consider subexponential wild-type growth, surviving mutant clones will grow exponentially for large time, which can be unrealistic in some situations. Stochastic growth which accounts for environmental restrictions, for instance the logistic branching process, introduces further technical difficulties and is left for future work. We do note that, despite the drawbacks of deterministic mutant growth as discussed in Section \ref{alt}, when both the wild-type and mutant populations grow deterministically as $\tau^{\rho}$, it is easy to see that for large times the clone size distribution still displays a power-law tail,
$
 \lim_{t\to\infty} t f_{Y_t}(y) = \frac{\rho+1}{\rho}y^{1/\rho-1}.
$

An underlying motivation for this work is the scenario of primary tumours spawning metastases in cancer. We test our hypothesis regarding a power-law tail in metastasis size distributions by analysing empirical data. For 19 of 21 data-sets the power-law distribution is deemed more likely than an exponentially decaying distribution. The exponent of the power-law decay was estimated in each case and found to lie between -1 and -2. Interpreting this in light of our theory, either the primary tumour had entered a subexponential growth phase or, if one assumes exponential primary growth, the metastatic cells had a fitness advantage compared to those in the primary. Either way we can conclude that, for the majority of patients, the metastases grew faster than the primary tumour.  

%%%%%%%%%%%%%%%%%%%%%%%%%%%%%%%%  
\section*{Acknowledgments}
We thank Peter Keller,  Paul Krapivsky, Martin Nowak, Karen Ogilvie, Bartlomiej Waclaw and Bruce Worton for helpful discussions. MDN acknowledges support from EPSRC via a studentship.

%%%%%%%%%%%%%%%%%%%%%%%%%%%%%%%%  
%%%%%%%%%%%%%%%%%%%%%%%%%%%%%%%%  
\appendix

%%%%%%%%%%%%%%%%%%%%%%%%%%%%%%%%  
 \section{Special functions, definitions and requisite results}\label{specialfunctions}
 
 Required definitions and identities taken from \cite{NIST:DLMF} unless otherwise stated.

 %%%%
 \begin{comment}
 Euler's reflection formulas for the gamma function are
 \begin{gather}
 \label{Euler1}
 \Gamma(1+z)\Gamma(1-z)=\frac{\pi z}{\sin(\pi z)} \quad \text{and} \quad \Gamma(z+1)\Gamma(z-1)=\frac{z}{z-1}\Gamma(z)^{2}.
 \end{gather}
 \end{comment}
 %%%%
 
\indent  With $s,z\in\mathbb{C}$ the polylogarithm of order $s$ is defined as
\begin{align}
\mathrm{Li}_{s}(z)=\sum_{k\geq 1}\frac{z^{k}}{k^{s}}.
\end{align}
Note that $\mathrm{Li}_{1}(z)=-\log(1-z)$. A required identity (from \cite{Weisstein:2015}) is
\begin{gather}\label{poly1}
\mathrm{Li}_{-n}(z)=\sum_{k=0}^{n}k! S(n+1,k+1)\bigg(\frac{z}{1-z}\bigg)^{k+1}
\end{gather}
for $n\in\mathbb{N}$. Here $S(n,k)$ are the Stirling numbers of the second kind.

%%%%
\begin{comment}
 The Hurwitz zeta function is defined for $a\in \mathbb{C}\setminus \mathbb{Z}_{\leq0}$ as
\begin{gather}
\zeta(s,a)=\sum_{k\geq0}\frac{1}{(k+a)^{s}}
\end{gather}
 for $\textrm{Re}[s]>1$ and by analytic continuation to other $s\neq1$.
\end{comment}
%%%%

Gauss's hypergeometric function also appears and for complex $a,b,c,z$ is defined by the power series
\begin{gather}
\Fze{a,b}{c}z=\sum_{k\geq0}\frac{(a)_{k}(b)_{k}}{(c)_{k}}\frac{z^{k}}{k!}\quad \text{for} \quad |z|<1,
\end{gather}
and by analytic continuation elsewhere. Here $(a)_{k}$ denotes the Pochhammer symbol or rising factorial, that is
\begin{gather}
(a)_{k}=\frac{\Gamma(a+k)}{\Gamma(a)}=a(a+1)(a+2)\cdots(a+k-1).
\end{gather}
 Some required identities for the hypergeometric function are:
 \begin{gather}
 \label{hypgeom1}
 \Fze{a,b}{c}z=(1-z)^{-b}\Fze{c-a,b}{c}{\frac{z}{z-1}},
 \end{gather}

 %%%%
 \begin{comment}
  \begin{multline}
 \label{hypgeom2}
 \Fze{a,b}{c}z=\frac{\Gamma(c)\Gamma(c-a-b)}{\Gamma(c-b)\Gamma(c-a)}\Fze{a,b}{a+b-c+1}{1-z}
 \\
 +\frac{\Gamma(c)\Gamma(a+b-c)}{\Gamma(a)\Gamma(b)}z^{1-c}(1-z)^{c-a-b}\Fze{1-a,1-b}{c-a-b+1}{1-z},
 \end{multline}
\end{comment}
%%%%% 

\begin{equation}\label{hypgeom3}
\Fze{1,b}{c}z=1+\frac{b}{c}z\Fze{1,b+1}{c+1}z,
\end{equation}

\begin{equation}\label{hypgeom4}
\Fze{1,1}{2}z=-\frac{\log(1-z)}{z},
\end{equation}
and the following connection can be made to the incomplete beta-function
\begin{equation}\label{hypgeom5}
\frac{z^{a}}{a}\Fze{a,1-b}{a+1}z=B_{z}(a,b)=\int_{0}^{x}t^{a-1}(1-t)^{b-1}\,dt.
\end{equation}

For any analytic function $f(z)=\sum_{n\geq 0}f_{n}z^{n}$, we denote the n$t$h coefficient as
$
[z^{n}]f(z)=f_{n}.
$
\begin{theorema}[\cite{Flajolet:2009}: Exponential Growth Formula]
  \label{Flajoletthm0}
If $f(z)$ is analytic at $0$ and $R$ is the modulus of a singularity nearest the origin in the sense that 
$
R:=\sup\{r\geq 0 |\text{ f is analytic in }|z|<r\}.
$
Then the coefficient $[z^{n}]f(z)$ satisfies
$
f_{n}=R^{-n}\Theta(n)
$
where $\limsup_{n}\sqrt[n]{|\Theta(n)|}=1$.
\end{theorema}

%%%%
\begin{comment}
    \begin{theorema}[\cite{Flajolet:2009}: Standard function scale]
  \label{Flajoletthm1}
    Let $a$ be an arbitrary complex number in $\mathbb{C}\setminus\mathbb{Z}_{\leq 0}$. The coefficient of $z^{n}$ for large $n$ has a full asymptotic expansion in descending powers of $n$
    \begin{equation}
    \begin{split}
    [z^{n}](1-z)^{-a}& \sim \frac{n^{a-1}}{\Gamma(a)}\bigg(1+\sum_{k\geq 1}\frac{e_{k}}{n^{k}}\bigg)
    \\
   & \sim \frac{n^{a-1}}{\Gamma(a)}\bigg(1+\frac{a(a-1)}{2n}+\frac{a(a-1)(a-2)(3a-1)}{24 n^2}+\ldots
    \end{split}
    \end{equation}
where $e_{k}$ is a polynomial of degree $2k$.
\end{theorema}

\begin{theorema}[\cite{Flajolet:2009}: Transfer, Big-Oh and little-oh]
\label{Flajoletthm2}
Let $a,\,b$ be arbitrary real numbers and $f(z)$ a function that is $\Delta$-analytic. \\
\indent (i) Assume that $f(z)$ satisfies in the intersection of a neighbourhood of 1 with its $\Delta$-domain the condition
\begin{gather}
f(z)=O\bigg((1-z)^{-a}(\log(\frac{1}{1-z})^{b}\bigg)
\end{gather}
Then one has: $\quad [z^{n}]f(z)=O(n^{a -1}(\log n)^{b})$. 
\end{theorema}
\end{comment}
%%%%
%The requirement that $f(z)$ is $\Delta$-analytic is a technicality that will hold in our analysis. 

We utilise several results from \cite{Bingham:1987} on the theory of regularly varying functions which we now define.
\begin{definition}\label{regvar}\cite{Bingham:1987}
A Lebesgue measurable function $f:\mathbb{R}^+\mapsto\mathbb{R}$ that is eventually positive is \textit{regularly varying} (at infinity) if for some $\kappa\in \mathbb{R}$,
\begin{equation}
 \lim_{t\rightarrow \infty} \frac{f(tx)}{f(t)}=x^{\kappa}, \quad x>0.
\end{equation}
The notation $f\in RV_{\kappa}$ will be used and we will denote $f\in RV_{0}$ as \textit{slowly varying functions}.
 \end{definition}
\begin{theorema}[\cite{Bingham:1987}: Characterisation Theorem]\label{Karamata2}
Suppose $f:\mathbb{R}^{+}\mapsto \mathbb{R}$ is measurable, eventually positive, and
$
\lim_{t\rightarrow \infty} \frac{f(tx)}{f(t)}
$
 exists, and is finite and positive for all $x$ in a set of positive Lebesgue measure. Then, for some $\kappa \in \mathbb{R}$, 
 \begin{enumerate}[(i)]
 \item $ f\in RV_{\kappa}$.
 \item $f(y)=y^{\kappa}l(y)$ where $l\in RV_{0}$. 
\end{enumerate}
\end{theorema}
\begin{propa}[\cite{Bingham:1987}: Proposition 1.3.6]\label{logprop}
For $f\in RV_{\kappa}$
$
\lim_{t\rightarrow\infty}\frac{ \log f(t)}{\log t }= \kappa.
$
\end{propa}
  \begin{theorema}[\cite{Bingham:1987}: Karamata's Theorem]\label{Karamata1}
     If $f\in RV_{\kappa}$, $X$ sufficiently large such that $f(y)$ is locally bounded in $[X,\infty)$, and $\kappa>-1$, then
\begin{equation}
\int_{X}^{y}f(s)\,ds\sim \frac{yf(y)}{\kappa+1}\quad \text{as} \quad y\rightarrow \infty.
\end{equation}
\end{theorema}

\begin{propa}[\cite{Bingham:1987}: Proposition 1.5.9.a]\label{slowprop}
Let $l\in RV_{0}$ and choose $X$ so that $l$ is locally integrable on $[X,\infty)$\begin{comment}(measurable and locally bounded is sufficient)\end{comment}
. Then 
\begin{enumerate}[(i)]
\item $\int_{X}^{x}\frac{l(t)}{t}\,dt \in RV_{0} $ .
\item $\frac{1}{l(x)}\int_{X}^{x}\frac{l(t)}{t}\,dt\rightarrow \infty\text{ as } x\rightarrow \infty$.
\end{enumerate}
\end{propa}
%\mn{Currently we don't know the rate of divergence of the fraction on the right, however the result may be circa page 164 in Bingham}

\section{Proofs for Section \ref{Connection}}\label{appendtheory}
In this work we have fixed the birth rate to be one. Other works, for example \cite{Keller:2015}, use a birth-death process with birth rate $\alp'$ and death rate $\beta'$ under timescale $t'$. Then the timescale used in the present work is defined by $t=\alp' t'$. This in turn implies that all rates under $t$ are given by dividing the corresponding rate under $t'$ by $\alp'$, e.g. $\beta=\frac{\beta'}{\alp'}$.

\begin{lemmaa}\label{recursion}
Consider generating functions $F(s)=\sum_{n\ge 0}p_n s^n$ and $G(s)=\sum_{n\ge 0}q_n s^n$ where $F(s)=e^{G(s)}$. Then $p_0=e^{q_0}$ and for $n\ge 1$ the following recursion holds
$$
 np_n = \sum_{k=0}^{n-1} (n-k) p_k q_{n-k} \ .
$$
\end{lemmaa}
\begin{proof}
Clearly $p_0=e^{q_0}$ from $F(0)=e^{G(0)}$. By differentiating $F(s)$ we obtain $F'(s)= F(s) G'(s)$, and in general
$$
 F^{(n)}(s) = \sum_{k=0}^{n-1} \binom{n-1}{k} F^{(k)}(s) G^{(n-k)}(s)
$$
which can be shown by induction using Pascal's formula for binomial coefficients. Evaluating the above equation at $s=0$ and using that $F^{(m)}(0)=m! p_n$ and $G^{(m)}(0)=m! q_n$ we arrive at the announced recursion. 
\end{proof}

%Here we give a brief note on converting our results to those which have an arbitrary birth rate. Generally for a birth-death process we have a birth rate $\alp'$ and death rate $\beta'$ under timescale $t'$. This is the timescale used in \cite{Keller:2015} for example. Then the timescale used in the present work is defined by $t=\alp' t'$. This in turn implies that all rates under $t$ are given by dividing the corresponding rate under $t'$ by $\alp'$, for example
%\begin{equation}
%\beta=\frac{\beta'}{\alp'},\, \mu=\frac{\mu'}{\alp'},\,\alp=\frac{\alp'}{\alp'}=1.
%\end{equation}
\begin{comment}
Then the forward Kolmogorov equation \eqref{FK} would read
\begin{equation}
\partial_{t'}  \mathbb{P}(Z_{t'}=k)=\alp' (k-1) \mathbb{P}(Z_{t'}=k-1)+\beta'(k+1)  \mathbb{P}(Z_{t'}=k+1)-(\alp'+\beta') k \mathbb{P}(Z_{t'}=k)
\end{equation}
\end{comment}
\begin{proof}[Proof of Proposition \ref{smallarrive}]
Utilising generating functions,
\begin{equation}
\begin{split}
\mathbb{E}(s^{B_{t}}|B_{t}>0) &=\frac{\gentotal(s)-\gentotal(0)}{1-\gentotal(0)} 
=\frac{e^{\mathbb{E}(K_{t}) (\genclone(s)-1)}-e^{\mathbb{E}(K_{t}) (\genclone(0)-1)}}{1-e^{\mathbb{E}(K_{t}) (\genclone(0)-1)}}
\\
& =\frac{\mathbb{E}(K_{t}) (\genclone(s)-1)-(\genclone(0)-1))}{-\mathbb{E}(K_{t}) (\genclone(0)-1)}+O(\mathbb{E}(K_{t}) )
\\
& =\frac{\genclone(s)-\genclone(0)}{1-\genclone(0)}+O(\mathbb{E}(K_{t}) )=\mathbb{E}(s^{Y_{t}}|Y_{t}>0)+O(\mathbb{E}(K_{t}) ).
\end{split}
\end{equation}
 \end{proof}

%%%%%%%%%%%%%%%%%%%%%%%%%%%%%%%%  
\section{Proofs for Section \ref{clonesize}}\label{clonalappend}

    We derive the generating function for the clone size distribution for stochastic growth and power-law wild-type growth, $n_{\tau}=\tau^{\rho}$, given in \eqref{powergf}. From \eqref{SinglecloneGF} we have
 \begin{gather}
 \genclone(s)=\frac{\rho+1}{t^{\rho +1}}\int_{0}^{t}\tau^{\rho}\bigg(1-\frac{\lambda}{1-\xi e^{-\lambda (t-\tau)}}\bigg)\, d\tau
 =1-\frac{(\rho +1)\lambda}{t^{\rho +1}}\int_{0}^{t}\frac{\tau^{\rho}}{1-\xi e^{-\lambda (t-\tau)}}\,d\tau.
 \end{gather}
It is enough to show
 \begin{gather}\label{toshow}
  \int \frac{\tau^{\rho}}{1-\xi e^{-\lambda (t-\tau)}}\,d\tau=  \frac{\tau^{\rho+1}}{\rho+1}+\rho !
 \sum_{i=0}^{\rho}\frac{(-1)^{i}}{(\rho -i)!\lambda^{i+1}}\tau^{\rho -i}\mathrm{Li}_{i+1}(\xi e^{-\lambda (t-\tau)})+C
 \end{gather}
 where $C$ is a constant of integration. This may be derived by a binomial expansion of the denominator and an identity for the incomplete gamma function, but for succinctness we simply differentiate both sides with respect to $\tau$.  First we note that
\begin{gather}
z \partial_{z} \mathrm{Li}_{i}(z)=\mathrm{Li}_{i-1}(z).
\end{gather}
Now differentiating the right hand side of \eqref{toshow} yields
\begin{gather}
\tau^{\rho} +\frac{\tau^{\rho}\lambda \mathrm{Li}_{0}(\xi e^{-\lambda (t-\tau)})}{\lambda}+
\rho ! \sum_{j=0}^{\rho-1} \frac{(-1)^{j}(\rho-j)\tau^{\rho-j-1}\mathrm{Li}_{j+1}(\xi e^{-\lambda (t-\tau)})}{(\rho-j)! \lambda^{j+1}} +
 \notag
\\ \rho !\sum_{i=1}^{\rho}\frac{(-1)^{i} \tau^{\rho-i}\lambda \mathrm{Li}_{i}(\xi e^{-\lambda(t-\tau)})}{(\rho-i)!\lambda^{i+1}}=\tau^{\rho}(1+\mathrm{Li}_{0}(\xi e^{-\lambda (t-\tau)}))
\end{gather}
where the equality follows by the telescoping nature of the sums. Noting that $(1-\xi e^{-\lambda (t-\tau)})^{-1}=\mathrm{Li}_{0}\left(\xi e^{-\lambda (t-\tau)}\right)+1
$ and applying the limits of the integral gives the desired result.
 
To determine the mass function, we seek a power series representation of the generating function. We focus on the $\beta=0$ case and thus $\xi=\frac{s}{s-1}$. By the definition of the polylogarithm and the binomial theorem
\begin{gather}
\mathrm{Li}_{i}\left(\frac{s}{s-1}\right)=\sum_{k\geq1}\sum_{j\geq 0}{k+j-1\choose j}(-1)^{k}\frac{s^{j+k}}{k^i}.
\intertext{Reindexing the sum we obtain}
\mathrm{Li}_{i}\left(\frac{s}{s-1}\right)=\sum_{m\geq 1}s^{m}\sum_{k=1}^{m}{m-1\choose m-k}\frac{(-1)^k}{k^i}\quad \text{and}\quad \mathrm{Li}_{i}\left(\frac{se^{-\alpha t}}{s-1}\right)=\sum_{m\geq 1}s^{m}\sum_{k=1}^{m}{m-1\choose m-k}\frac{(-e^{-\alpha t})^k}{k^i}.
\end{gather}
Applying this to the polylogarithmic terms in $\genclone(s)$, and noting
\begin{gather}
\sum_{k=1}^{m}{m-1\choose m-k}\frac{(-1)^k}{k^i}=\frac{1}{m}\sum_{k=1}^{m}{m\choose k}\frac{(-1)^k}{k^{i-1}}\quad \text{and}\quad \sum_{k=1}^{m}{m\choose k}(-1)^k=-1,
\end{gather}
yields  \eqref{powerpmf} as the desired mass function.

\begin{proof}[Proof of Proposition \ref{pmfshape}]
   Using \eqref{Singleclonepmf}, we see that for $k\geq 1$
\begin{align}
\mathbb{P}(Y_{t}=k+1)-\mathbb{P}(Y_{t}=k)=\frac{1}{a_t}\int_{0}^{t}n_{t-\tau}\left[\mathbb{P}(Z_{\tau}=k+1)-\mathbb{P}(Z_{\tau}=k)\right]\,d\tau.
\end{align}
Now from \eqref{BDpmf} it is clear that the integrand is negative for finite, positive $\tau$ giving the result.
\end{proof}

\begin{proof}[Proof of Theorem \ref{cutoff}]
The result is an application of Theorem \ref{Flajoletthm0}. We seek the closest to the origin singularity of
\begin{equation}
I_{t}(s)=\int_{0}^{t}n_{\tau}\mathcal{Z}_{t-\tau}(s)\,d\tau=\int_{0}^{t}n_{t-\tau}\mathcal{Z}_{\tau}(s)\,d\tau
\end{equation}
which is claimed to be at $\mathcal{S}_{t}$. Indeed, we note that for $|s|<\mathcal{S}_{t}$, $\mathcal{Z}_{\tau}(s)$ is analytic for all $\tau$ and as $n_{\tau}$ is continuous we can conclude that the $I_{t}(s)$ is analytic in this region also (Chapter 2, Theorem 5.4 in \cite{Stein:03}). As $n_{\tau}>0$ there exists $\epsilon>0$ such that
\begin{equation}
|I_{t}(s)|\geq \epsilon \bigg|\int_{0}^{t}\mathcal{Z}_{\tau}(s)\,d\tau\bigg|=\epsilon \bigg|\beta t+\log\bigg[\frac{\lambda}{1-\beta e^{-\lambda t}-s(1-e^{-\lambda t})}\bigg]\bigg|.
\end{equation}
The rightmost expression can be seen to have closest to origin singularity at $\mathcal{S}_{t}$ and as $a_{t}\genclone(s)=I_{t}(s)$, by Theorem \ref{Flajoletthm0}, we can conclude Theorem \ref{cutoff}.
 
\end{proof}

%%%%%%%%%%%%%%%%%%%%%%%%%%%%%%%%  
\section{Proofs for Section \ref{longtime}}\label{unproofs}

\begin{proof}[Proof of Lemma \ref{limform}]

Choose $x\geq 0$ and let $y=e^{t},\,c=e^{-x}$. Consider the function $g(z)=n_{\log(z)}$. Then Theorem \ref{Karamata2}(i) yields
\begin{equation}
\lim_{t\rightarrow\infty}\frac{n_{t-x}}{n_{t}}=\lim_{y\rightarrow\infty}\frac{g(yc)}{g(y)}=c^{\delta^*}=e^{-x\delta^*}.
\end{equation}
Further, Proposition \ref{logprop} gives
\begin{equation}
\lim_{y\rightarrow\infty}\frac{\log g(y)}{\log y }=\lim_{t\rightarrow \infty}\frac{\log(n_{t})}{t}= \delta^*\geq 0.
\end{equation}
The non-negativity of $\delta^*$ is dictated by the monotone increasing nature of $n_{\tau}$.
\end{proof}

To prove Theorem \ref{meanlim} we require the following:

\begin{lemmaa}\label{subexp1}
Let $s_{1}(t),\,s_{2}(t)$ be subexponential functions, then
\begin{enumerate}[(i)]
 \item $n_{t}=e^{t\delta^* }s_{1}(t)$.
 \item For $\eta\geq 0,\,C>0$ 
 \begin{equation}
 \int_{0}^{t}n_{\tau}e^{-\eta \tau}\,d\tau \sim\begin{cases}
  \frac{e^{(\delta^*-\eta ) t}s_{1}(t)}{\delta^*-\eta} & \eta<\delta^*  \\
  s_{2}(t) & \delta^*=\eta \\
  C & \delta^*<\eta
  \end{cases}
  \quad\text{as } t\rightarrow\infty.
 \end{equation}
 \end{enumerate}
\end{lemmaa}
 We highlight that neither subexponential function depend on $\eta$.
\begin{proof}
(i) For $y=e^{t}$, $n_{\log y}=g(y)$. Now $g\in RV_{\delta^*}$ hence $g(y)=y^{\delta^*}l(y)$ with $l\in RV_{0}$ by Theorem \ref{Karamata2}(ii). Setting $s_1(\log y)=l(y)$, by Lemma \ref{limform}, $s_{1}(t)$ is subexponential.
 (ii) Let $g(y)$ be as above. With $\delta^* > \eta\geq 0$ and using the change of variables $s=\log \tau$ we have
\begin{equation}\label{intconv}
\int_{0}^{t}n_{\tau}e^{-\eta\tau}\,d\tau=\int_{1}^{y}g(s)s^{-1-\eta}\,ds \sim \frac{y^{-\eta}g(y)}{\delta^*-\eta}=\frac{e^{(\delta^*-\eta)t}s_1(t)}{\delta^*-\eta}
\end{equation}
where the asymptotic equivalence is due to Theorem \ref{Karamata1} applied to $g(y)y^{-1-\eta}\in RV_{\delta^*-\eta-1}$ and the final equality is by part (i). For $\delta^*=\eta$, by Theorem \ref{Karamata2}(ii) the integrand will be a subexponential function. Applying the same change of variables as in \eqref{intconv} we see by Proposition \ref{slowprop}(i) that the integral is a slowly varying function in $y$ and hence is subexponential in $t$, which we denote $s_{2}(t)$. Now for $\delta^*<\eta$, by Lemma \ref{limform} we may choose $t$ large enough such that 
$
n_{t}^{1/t}<e^{(\delta^*+\eta)/2}
$
which by a basic result for Laplace transforms, see, e.g., Theorem 1.11 in \cite{Schiff:1999}, ensures convergence to a finite, positive constant.
\end{proof}
As an example, which will be useful for the next proof, we apply the above lemma to $a_{t}$. With $s_{1}(t),\,s_{2}(t)$ subexponential functions
\begin{equation}
a_{t}=\int_{0}^{t}n_{\tau}\,d\tau\sim
\begin{cases}
\frac{e^{\delta^* t}s_{1}(t)}{\delta^*} &  \delta^* >0 \\
s_{2}(t) & \delta^*=0  
\end{cases}
\quad \text{as }t\rightarrow \infty.
\end{equation}

\begin{proof}[Proof of Theorem \ref{meanlim}]
 We require the first and second moments of $Z_{t}$ which may be found by differentiating \eqref{BDgf}, or see Lemma \ref{BDlemma}. Then
\begin{equation}\label{mean}
\mathbb{E}(Y_{t})=\frac{1}{a_{t}}\int_{0}^{t}n_{\tau}\mathbb{E}(Z_{t-\tau})\,d\tau=e^{\lambda t}\frac{\int_{0}^{t}n_{\tau}e^{-\lambda \tau}\,d\tau}{\int_{0}^{t}n_{\tau}\,d\tau},
\end{equation}
\begin{equation}\label{var}
\mathbb{E}(Y_{t}^2)=\frac{1}{a_{t}}\int_{0}^{t}n_{\tau}\mathbb{E}(Z_{t-\tau}^2)\,d\tau=\frac{e^{2\lambda t}}{a_{t}\lambda}\left(2\int_{0}^{t}n_{\tau}e^{-2\lambda \tau}\,d\tau-(2-\lambda)e^{-\lambda t}\int_{0}^{t}n_{\tau}e^{-\lambda \tau}\,d\tau\right)  .
\end{equation}
Throughout let $s_{t}$ be a generic subexponential function and it will be helpful to observe that the reciprocal or constant multiples of a subexponential function are subexponential. We first consider the mean. For the cases $\delta^*\neq\lambda$, applying Lemma \ref{subexp1}(ii) to \eqref{mean} with $\eta=\lambda$ for the numerator and $\eta=0$ for the denominator proves the claim. For $\delta^*=\lambda$, using Lemma \ref{subexp1}(i) then (ii), we have
\begin{equation}\label{meanasy}
\mathbb{E}(Y_{t})=\frac{e^{\delta^* t}\int_{0}^{t}e^{\delta^* \tau}s_{\tau}e^{-\delta^*\tau}}{\int_{0}^{t}e^{\delta^*\tau}s_{\tau}\,d\tau}\sim\frac{\delta^*\int_{0}^{t}s_{\tau}\,d\tau}{s_{t}}=s_{1}(t).
\end{equation}
That $s_{1}(t)$ diverges can be seen by applying the standard change of variables $t=\log(y),\,\tau=\log(s)$ coupled with Proposition \ref{slowprop}(ii). 
 Turning to the variance, with $\textrm{Var}(Y_{t})=\mathbb{E}(Y_{t}^2)-\mathbb{E}(Y_{t})^2$, when $\delta^*>2\lambda$ we apply Lemma \ref{subexp1}(ii) term by term to \eqref{var}. For $\delta^*<\lambda$ all integrals converge and so, with $C_{1},\,C_{2}$ constants,
 \begin{equation}
 \mathrm{Var}(Y_t)\sim \frac{e^{2\lambda t}}{a_{t}}\left(C_{1}-\frac{C_{2}}{a_{t}}\right)\sim C_{1}\frac{e^{2\lambda t}}{a_{t}}.
 \end{equation} 
 The last relation is due to the monotonicity of $a_{t}$ and the desired representation is obtained by applying Lemma \ref{subexp1}(ii) to $a_{t}$ and absorbing $C_{1}$ into $s_{4}(t)$. When $\lambda\leq\delta^*<2\lambda$ the same argument holds as long as we note that
\begin{equation}
e^{-\lambda t}\int_{0}^{t}n_{\tau}e^{-\lambda t}\,d\tau=e^{-\lambda t}\int_{0}^{t}e^{(\delta^*-\lambda )\tau}s_{\tau}\,d\tau\leq e^{-(2\lambda-\delta^*)t}\int_{0}^{t}s_{\tau}\,d\tau.
\end{equation}
By Proposition \ref{slowprop}(i) the rightmost integral is a subexponential function and as we may always choose $t$ sufficiently large such that $s_{t}^{1/t}<e^{2\lambda-\delta^*}$, we find
\begin{equation}
e^{-\lambda t}\int_{0}^{t}n_{\tau}e^{-\lambda t}\,d\tau \rightarrow 0.
\end{equation}
Applying Lemma \ref{subexp1} to $ a_{t}^{-1}\int_{0}^{t}n_{\tau}e^{-\lambda t}\,d\tau$ demonstrates the contribution from the mean squared is negligible. 
For $\delta^*=2\lambda$ we apply the same argument as in \eqref{meanasy} to each term and this completes the proof.
\end{proof}

In order to prove Theorem \ref{bigthm} we require the following lemma.
\begin{lemmaa}\label{todominate}
For $|s|<1,\,\beta\in[0,1)$, $u\in[0,1]$ and $\xi$ as in \eqref{BDgf}, we have 
\begin{equation}
\left|\frac{\xi}{1-\xi u}\right|\leq \left|\frac{\beta-s}{1-\max\{\beta,|s|\}}\right|.
\end{equation}
\end{lemmaa}
\begin{proof}
By the definition of $\xi$, 
\begin{equation}
\frac{\xi}{1-\xi u}=\frac{\beta-s}{1-s-(\beta-s)u}.
\end{equation}
The triangle inequality yields
\begin{equation}
|1-s-\beta u+su|=|1-(\beta u+s(1-u))|\geq |1-|\beta u+s(1-u)||
\end{equation}
and
\begin{equation}
|\beta u+s(1-u)|\leq u\beta+(1-u)|s|\leq \max\{\beta,|s|\}.
\end{equation}
The claimed inequality now follows.
\end{proof}

\begin{proof}[Proof of Theorem \ref{bigthm}]
To avoid division by $0$ let $t>0$. Taking the generating function for $Y_{t}$ from equation \eqref{SinglecloneGF} we apply the change of variables $u=e^{-\lambda \tau}$ which gives
\begin{equation}
\genclone(s)-\beta=-\frac{1}{a_{t}}\int_{e^{-\lambda t}}^{1}n_{t+\frac{\log u}{\lambda}}\frac{\xi}{1-\xi u}\,du.
\end{equation}
Now recalling $n_{\tau}=0$ for $\tau<0$ and multiplying both sides by $\frac{a_t}{n_t}$ yields
\begin{equation}
\frac{a_{t}}{n_t}(\genclone(s)-\beta)=-\int_{0}^{1}\frac{n_{t+\frac{\log u}{\lambda}}}{n_{t}}\frac{\xi}{1-\xi u}\,du.
\end{equation}
Noting that by monotonicity
$
n_{t+\frac{\log u}{\lambda}}\big/n_{t}\leq 1,
$
which coupled with Lemma \ref{todominate} shows the integrand may be dominated. By assumption the integrand converges and therefore, using Lemma \ref{limform} and the dominated convergence theorem, we have 
\begin{equation}
\lim_{t\rightarrow\infty}\frac{a_{t}}{n_t}(\genclone(s)-\beta)=-\int_{0}^{1}u^{\delta^*/\lambda}\frac{\xi}{1-\xi u}\,du
=\frac{-1}{\xi^{\gamma^*}}\mathrm{B}_{\xi}(\gamma^*+1,0)=
\\
\frac{1}{\gamma^*}\left[1-\Fze{1,\gamma^*}{1+\gamma^*}{\xi}\right].
\end{equation}
The final equality follows from applying \eqref{hypgeom5} then \eqref{hypgeom3}.
\end{proof}

\begin{proof}[Proof of Corollary \ref{bigcor}]
The first statement is given by applying Lemma \ref{subexp1}(ii) to $a_t/n_t$. Then taking the limiting generating function in Theorem \ref{bigthm} we firstly apply \eqref{hypgeom3} then \eqref{hypgeom4} which yields generating function representation for $\gamma^*=0$. The mass function for $\gamma^*=0$ is simply a logarithmic expansion. For $\gamma^*>0$ we use the expression given in Appendix A of \cite{Keller:2015} to obtain a series expansion for the generating function in terms of $s$ and the coefficients of the expansion give the mass function.
\end{proof}

\begin{proof}[Proof of Corollary \ref{tailcor}] 
The analysis involves expanding the limiting generating function in Theorem \ref{bigthm} around it's singularity at $s=1$ and exactly mirrors that given in section 6 of \cite{Keller:2015} and so is omitted.

\end{proof}

\begin{proof}[Proof of Theorem \ref{largepopsmallmut}]
The mean and variance can be obtained by using \eqref{meanlink} with Theorem \ref{meanlim} (the second moment dominates the mean squared in all divergent cases). For the generating function, \eqref{GFlink} gives 
\begin{equation}
\gentotal(s)=\exp\left[\mu a_{t}(\genclone(s)-1)\right]=\exp\left[\theta \frac{a_{t}}{n_{t}}(\genclone(s)-\beta +\beta-1)\right]
\end{equation}
which coupled with Theorem \ref{bigthm} yields the result. The map in \eqref{GFlink} is analytic so the tail can be obtained by its expansion coupled with Corollary \ref{tailcor}.
\end{proof}

%%%%%%%%%%%%%%%%%%%%%%%%%%%%%%%%  
\section{Maximum likelihood estimators for distributions considered}\label{MLEapp}
Consider a data-set $\textbf{y}=(y_{1},\ldots,y_{N})$ assumed to be a realisation of the random vector $\textbf{Y}_{t}=(Y^{(1)}_{t},\ldots,Y^{(N)}_{t})$ where all $Y_{t}^{(i)}$ are $\textit{iid}$ random variables representing the metastasis sizes and let $m=\text{min}(\textbf{y})$. Then the maximum likelihood estimator (MLE) $\hat \omega $ for a one parameter probability distribution $\mathbb{P}(\textbf{Y}_{t}=\textbf{y};\omega)$ is 
\begin{gather}
\label{Mle}
\hat\omega=\argmax_{\omega}\log(\LTL(\omega))=\argmax_{\omega}\log (\mathbb{P}(\textbf{Y}_{t}=\textbf{y};\omega))=\argmax_{\omega}\log\bigg[\prod_{i=1}^{N}\mathbb{P}(Y_{t}^{(i)}=y_{i};\omega)\bigg]
\end{gather}
where $\LTL(\omega)$ is the likelihood function and the joint distribution becomes a product by independence. We derive the MLE under the assumption that the data is sampled from a distribution whose tail follows the geometric distribution, the power-law case is analogous but for the final step. 
\\ \indent Assume for at least $y\geq m$, that $\mathbb{P}_{2}(Y_t=y;p)=C_{2}p(1-p)^{y}$. Let $A=\mathbb{P}(Y_{t}<m)$ (no indices are required as this quantity is assumed independent of the tail) then
\begin{equation}
\sum_{y\geq m} \pr _{2}(Y_t=y;p)=\sum_{y\geq m}C_{2}p(1-p)^y=1-A \implies C_{2}=(1-A)(1-p)^{-m}.
\end{equation}
The log-likelihood is now given by 
\begin{equation}
\log \LTL(p)=N\log(1-A)-mN\log(1-p)+N\log(p)+\log(1-p)\sum_{i=1}^{N}y_{i}.
\end{equation}
Setting $\partial_{p}\log \LTL(p)=0$ we solve to find the MLE
\begin{equation}
\hat p=\frac{N}{\sum_{i=1}^{N}y_{i}+N(1-m)}.
\end{equation}
The power-law case is analogous. There $C_{1}=\frac{1-A}{\zeta(\omega,m)}$, where $\zeta$ is the Hurwitz zeta function. No closed form expression is found for the MLE and instead we have the approximation \cite{Clauset:2009}
\begin{equation}
\hat\omega \approx 1+N\bigg[\sum_{i=1}^{N}\log\bigg(\frac{y_{i}}{m-\frac{1}{2}}\bigg)\bigg]^{-1}.
\end{equation}
This was used for the estimates given in Section \ref{applications}.

%%%%%%%%%%%%%%%%%%%%%%%%%%%%%%%%  
\section{Proofs for Section \ref{alt}}\label{altproofs}

In order to prove Theorem \ref{thmcums} we require the following lemma regarding the moments of the birth-death process:
 \begin{lemmaa}
\begin{gather}
\mathbb{E}(Z_{t}^{m})=\frac{\lambda}{(1-\beta e^{-\lambda t})}\sum_{k=0}^{m}k!S(m+1,k+1)\bigg(\frac{e^{\lambda t}-1}{\lambda}\bigg)^{k},
\end{gather}
where $S(m,k)$ are Stirling numbers of the second kind. Hence, as $t\rightarrow \infty$,
\begin{gather}
\mathbb{E}(Z_{t}^m)=m!e^{m\lambda t}\lambda^{-(m-1)}+O(e^{(m-1)\lambda t}).
\end{gather}
\label{BDlemma}
\end{lemmaa}
\begin{proof}[Proof of Lemma \ref{BDlemma}]
Recall the generating function for the birth-death process \eqref{BDgf} whose power series representation has coefficients given by \eqref{BDpmf}. The moment generating function is $M_{Z_{t}}(s)=\mathcal{Z}_{t}(e^{s})$ and hence
\begin{gather}
\mathcal{Z}_{t}(e^{s})-\mathcal{Z}_{t}(0)=\left(1-\beta \mathcal{S}_{t}^{-1}\right)(\mathcal{S}_{t}-1)\sum_{j\geq 1}\mathcal{S}_{t}^{-j}e^{sj}.
\end{gather}
Thus for $m\geq 1$
\begin{equation}
\mathbb{E}(Z_{t}^{m})=\partial_{s}^{m}\mathcal{Z}_{t}(e^{s})|_{s=0}=\left(1-\beta \mathcal{S}_{t}^{-1}\right)(\mathcal{S}_{t}-1)\sum_{j\geq 1}\mathcal{S}_{t}^{-j}j^{m}.
\end{equation}
Since $\sum_{j\geq 1}\mathcal{S}_{t}^{-j}j^{m}=\mathrm{Li}_{-m}(\mathcal{S}_{t}^{-1})$, we can use \eqref{poly1} and thus arrive at our first result. Note $S(m,m)=1$ and so focusing on the leading order in $t$, the summand with $k=m$ is $m!\bigg(\frac{e^{\lambda t}-1}{\lambda}\bigg)^{m}=m!e^{m\lambda t}\lambda^{-m}+O(e^{(m-1)\lambda t})$, which proves the claim.
\end{proof}

\begin{proof}[Proof of Theorem \ref{thmcums}]
In the deterministic case we have
\begin{gather}
\mathbb{E}((Y_{t}^{\mathrm{Det}})^{m})=\frac{1}{a_{t}}\int_{0}^{t}n_{\tau}e^{m\lambda(t-\tau)}\,d\tau.
 \end{gather}
The moments for stochastic mutant growth are obtained from the moment generating function $M_{Y_{t}}(s)=\genclone(e^{s})$.  The moments are therefore
  \begin{gather}
  \label{stochmoms}
 \mathbb{E}((Y_{t}^{\mathrm{Stoch}})^{m})=\partial_{s}^{m}M_{Y_{t}}(0)=\frac{1}{a_{t}}\int_{0}^{t}n_{\tau}\partial_{s}^{m}\mathcal{Z}_{t-\tau}(e^s)|_{s=0}\,d\tau.
\end{gather}
Hence using the second statement in Lemma \ref{BDlemma} we have
\begin{gather}
\frac{\mathbb{E}((Y_{t}^{\mathrm{Stoch}})^{m})}{\mathbb{E}((Y_{t}^{\mathrm{Det}})^{m})}=\frac{m!}{\lambda^{m-1}}+O(e^{-\lambda t})
\end{gather}
which is the desired result.
\end{proof}

\begin{proof}[Proof of Theorem \ref{Poislink}]
Immediately from \eqref{Singleclonepmf} we see
\begin{equation}
\pr(Y_t>xe^{\lambda t}| Y_t>\epsilon e^{\lambda t})=\frac{\pr(Y_t>xe^{\lambda t})}{\pr(Y_t>\epsilon e^{\lambda t})}=\frac{\int_{0}^{t}n_{t-\tau}\pr(Z_\tau>xe^{\lambda t})\,d\tau}{\int_{0}^{t}n_{t-\tau}\pr(Z_\tau>\epsilon e^{\lambda t})\,d\tau}.
\end{equation}
It is enough to examine the numerator. As, from \eqref{BDpmf},
\begin{equation}
\pr(Z_{t}>k)=(1-\beta \mathcal{S}_t^{-1})\mathcal{S}_t^{-k}, \quad k\geq 0 
\end{equation}
we have
\begin{equation}\label{split}
\int_{0}^{t}n_{t-\tau}\pr(Z_{\tau}>xe^{\lambda t})\,d\tau=\int_{0}^{t}n_{t-\tau}\mathcal{S}_{\tau}^{-\lfloor xe^{\lambda t}\rfloor}\,d\tau-\beta\int_{0}^{t}n_{t-\tau}\mathcal{S}_{\tau}^{-\lfloor xe^{\lambda t}\rfloor-1}\,d\tau
\end{equation}
Here $\lfloor a \rfloor$ denotes the integer part of $a$, and is necessary as $\pr(Z_{t}>k)$ is defined on the non-negative integers. Focusing on the first term from the right hand side of \eqref{split} and using the definition of $\mathcal{S}_{\tau}$ \eqref{ratio} gives
\begin{gather}\label{firterm}
\int_{0}^{t}n_{t-\tau}\exp\left(-\lfloor x e^{\lambda t} \rfloor (\log(1-\beta e^{-\lambda \tau})-\log(1-e^{-\lambda \tau}))\right)\,d\tau.
\end{gather}
Now we change variables to $s=xe^{\lambda (t-\tau)}$ and note that the resulting integrand can be dominated by $\frac{n_{\lambda^{-1}\log(s/x)}}{\lambda s}e^{\lambda(1-s)}$ which is integrable for all $n_{\tau}$ under consideration (by Laplace transform arguments  \cite{Schiff:1999}). Hence the by dominated convergence theorem we can conclude
\begin{equation}
\lim_{t\rightarrow\infty}\int_{0}^{t}n_{t-\tau}\mathcal{S}_{\tau}^{-\lfloor xe^{\lambda t}\rfloor}\,d\tau=\lambda^{-1}\int_{x}^{\infty}n_{\lambda^{-1}\log(s/x)}\frac{e^{-\lambda s}}{s}\,ds.
\end{equation}
The second integral from the right hand side of \eqref{split} can be treated analogously and yields an identical result with $\beta$ as a prefactor. Hence 
\begin{equation}
\lim_{t\rightarrow\infty}\int_{0}^{t}n_{t-\tau}\pr(Z_{\tau}>xe^{\lambda t})\,d\tau=\int_{x}^{\infty}n_{\lambda^{-1}\log(s/x)}\frac{e^{-\lambda s}}{s}\,ds=\mu^{-1}m(x,\infty),
\end{equation}
and the claimed result follows.
\end{proof}

\bibliographystyle{spbasic}
\bibliography{cancer4_bmb}

\begin{thebibliography}{47}
\providecommand{\natexlab}[1]{#1}
\providecommand{\url}[1]{{#1}}
\providecommand{\urlprefix}{URL }
\expandafter\ifx\csname urlstyle\endcsname\relax
  \providecommand{\doi}[1]{DOI~\discretionary{}{}{}#1}\else
  \providecommand{\doi}{DOI~\discretionary{}{}{}\begingroup
  \urlstyle{rm}\Url}\fi
\providecommand{\eprint}[2][]{\url{#2}}

\bibitem[{Angerer(2001)}]{Angerer:2001}
Angerer WP (2001) {An explicit representation of the Luria-Delbr\"{u}ck
  Distribution}. Journal of Mathematical Biology 42(2):145--174

\bibitem[{Antal and Krapivsky(2010)}]{Antal:2010}
Antal T, Krapivsky PL (2010) {Exact solution of a two-type branching process:
  clone size distribution in cell division kinetics}. J Stat Mech P07028 (7)

\bibitem[{Antal and Krapivsky(2011)}]{Antal:2011}
Antal T, Krapivsky PL (2011) {Exact solution of a two-type branching process:
  models of tumor progression}. J Stat Mech P08018

\bibitem[{Athreya and Ney(2004)}]{Athreya:2004}
Athreya KB, Ney PE (2004) {Branching Processes}. Dover Publications

\bibitem[{Bartlett(1955)}]{Bartlett:1955}
Bartlett M (1955) {An Introduction to Stochastic Processes}, 3rd edn. Cambridge
  University Press

\bibitem[{Bingham et~al(1987)Bingham, Goldie, and Teugels}]{Bingham:1987}
Bingham N, Goldie C, Teugels J (1987) {Regular Variation}. Cambridge University
  Press

\bibitem[{Bozic and Nowak(2014)}]{Bozic:2014}
Bozic I, Nowak M (2014) {Timing and heterogeneity of mutations associated with
  drug resistance in metastatic cancers}. Proceedings of the National Academy
  of Science USA 111(45)

\bibitem[{Bozic et~al(2013)Bozic, Reiter, Allen, Antal, Chatterjee, Shah, Moon,
  Yaqubie, Kelly, Le, Lipson, Chapman, Diaz, Vogelstein, and
  Nowak}]{Bozic:2013}
Bozic I, Reiter JG, Allen B, Antal T, Chatterjee K, Shah P, Moon YS, Yaqubie A,
  Kelly N, Le DT, Lipson EJ, Chapman PB, Diaz LA, Vogelstein B, Nowak MA (2013)
  {Evolutionary dynamics of cancer in response to targeted combination
  therapy}. eLife: e00747 (2)

\bibitem[{Burnham and Anderson(1998)}]{Burnham:1998}
Burnham K, Anderson D (1998) { Model selection and inference: a practical
  information-theoretic approach}. Springer

\bibitem[{Clauset et~al(2009)Clauset, Shalizi, and Newman}]{Clauset:2009}
Clauset A, Shalizi CR, Newman MEJ (2009) {Power-law distributions in empirical
  data}. SIAM Review 51:661--703

\bibitem[{Dewanji et~al(2005)Dewanji, Luebeck, and Moolgavkar}]{Dewanji:2005}
Dewanji A, Luebeck EG, Moolgavkar SH (2005) {A generalized Luria-Delbr\"{u}ck
  model}. Mathematical Biosciences 197(2):140--152

\bibitem[{Dewanji et~al(2011)Dewanji, Jeon, Mexa, and Luebeck}]{Dewanji:2011}
Dewanji A, Jeon J, Mexa R, Luebeck EG (2011) {Number and size distribution of
  colorectal adenomas under the multistage clonal expansion model of cancer}.
  PLoS Computational Biology 7(10)

\bibitem[{{Diaz Jr} et~al(2012){Diaz Jr}, Williams, Wu, Kinde, Hecht, Berlin,
  Allen, Bozic, Reiter, Nowak, Kinzler, Oliner, and Vogelstein}]{Diaz:2012}
{Diaz Jr} LA, Williams RT, Wu J, Kinde I, Hecht JR, Berlin J, Allen B, Bozic I,
  Reiter JG, Nowak MA, Kinzler KW, Oliner KS, Vogelstein B (2012) {The
  molecular evolution of acquired resistance to targeted EGFR blockade in
  colorectal cancers}. Nature 486:537--540

\bibitem[{{\relax DLMF}(2016)}]{NIST:DLMF}
{\relax DLMF} (2016) {NIST Digital Library of Mathematical Functions}.
  http://dlmf.nist.gov/, Release 1.0.11 of 2016-06-08,
  \urlprefix\url{http://dlmf.nist.gov/}, online companion to
  \cite{Olver:2010:NHMF}

\bibitem[{Durrett(1996)}]{Durrett:1996}
Durrett R (1996) {Probability: Theory and Examples}, 4th edn. Cambridge Series
  in Statistical and Probabilistic Mathematics, Cambridge University Press

\bibitem[{Durrett(2015)}]{Durrett:2015}
Durrett R (2015) {Branching Process Models of Cancer}, 1st edn. Stochastics in
  Biological Systems, Springer International Publishing

\bibitem[{Durrett and Moseley(2010)}]{Durrett:2010}
Durrett R, Moseley S (2010) {Evolution of resistance and progression to disease
  during clonal expansion of cancer}. Theoretical Population Biology
  77(1):42--48

\bibitem[{Flajolet and Sedgewick(2009)}]{Flajolet:2009}
Flajolet P, Sedgewick R (2009) {Analytic Combinatorics}. Cambridge University
  Press, Cambridge

\bibitem[{Foo and Michor(2014)}]{Foo:2014}
Foo J, Michor F (2014) {Evolution of acquired resistance to anti-cancer
  therapy}. Journal of Theoretical Biology 355:10--20

\bibitem[{Hanin et~al(2006)Hanin, Rose, and Zaider}]{Hanin:2006}
Hanin L, Rose J, Zaider M (2006) {A stochastic model for the sizes of
  detectable metastases}. Journal of Theoretical Biology 243(3):407--417

\bibitem[{Houchmandzadeh(2015)}]{Houchmandzadeh:2015}
Houchmandzadeh B (2015) {General formulation of Luria-Delbr\"{u}ck distribution
  of the number of mutants}. Physical Review E, 012719 92

\bibitem[{Hudson(1971)}]{Hudson:1971}
Hudson DJ (1971) {Interval estimation from the likelihood function}. Journal of
  the Royal Statistical Society Series B (Methodological) 33(2):256--262

\bibitem[{Iwasa et~al(2006)Iwasa, Nowak, and Michor}]{Iwasa:2006}
Iwasa Y, Nowak MA, Michor F (2006) {Evolution of resistance during clonal
  expansion}. Genetics 172(4):2557--2566

\bibitem[{Jeon et~al(2008)Jeon, Meza, Moolgavkar, and Luebeck}]{Jeon:2008}
Jeon J, Meza R, Moolgavkar SH, Luebeck EG (2008) {Evaluation of screening
  strategies for pre-malignant lesions using a biomathematical approach}.
  Mathematical Biosciences 213(1):56--70

\bibitem[{Karlin and Taylor(1981)}]{Karlin:1981}
Karlin S, Taylor HM (1981) {A Second Course in Stochastic Processes}. New York
  - San Francisco - London: Academic Press, Inc., a subsidiary of Harcourt
  Brace Jovanovich, Publishers. XVI

\bibitem[{Karlin and Taylor(1998)}]{Karlin:1998}
Karlin S, Taylor HM (1998) {An Introduction to Stochastic Modeling}, 3rd edn.
  Academic Press Inc.

\bibitem[{Keller and Antal(2015)}]{Keller:2015}
Keller P, Antal T (2015) {Mutant number distribution in an exponentially
  growing population}. J Stat Mech P01011 (1)

\bibitem[{Kendall(1948)}]{Kendall:1948}
Kendall DG (1948) {On some modes of population growth leading to R. A. Fisher's
  logarithmic series distribution}. Biometrika 35(1/2):6--15

\bibitem[{Kendall(1960)}]{Kendall:1960}
Kendall DG (1960) {Birth-and-death processes, and the theory of
  carcinogenesis}. Biometrika 47:13--21

\bibitem[{Kessler and Levine(2015)}]{Kessler:2015}
Kessler D, Levine H (2015) {Scaling solution in the large population limit of
  the general asymmetric stochastic Luria–Delbr\"{u}ck evolution process}.
  Journal of Statistical Physics 158,(4):783--805

\bibitem[{Kessler et~al(2014)Kessler, Austin, and Levine}]{Kessler:2014}
Kessler DA, Austin RH, Levine H (2014) {Resistance to chemotherapy: patient
  variability and cellular heterogeneity.} Cancer research 74(17):4663--70

\bibitem[{Komarova et~al(2007)Komarova, Wu, and Baldi}]{Komarova:2007}
Komarova NL, Wu L, Baldi P (2007) {The fixed-size Luria-Delbr\"{u}ck model with
  nonzero death rate}. Mathematical Biosciences 210:253--290

\bibitem[{Krapivsky and Redner(2001)}]{Krapivsky:2001}
Krapivsky PL, Redner S (2001) {Organization of Growing Random Networks}.
  Physical Review E 1--066123 63(6)

\bibitem[{Lea and Coulson(1949)}]{Lea:1949}
Lea DE, Coulson CA (1949) {The distribution of the numbers of mutants in
  bacterial populations}. Journal of Genetics 49(3):264--285

\bibitem[{Luria and Delbr\"{u}ck(1943)}]{Luria:1943}
Luria SE, Delbr\"{u}ck M (1943) {Mutations of bacteria from virus sensitivity
  to virus resistance}. Genetics 48(6):491--511

\bibitem[{Murray(2002)}]{Murray:2002}
Murray JD (2002) {Mathematical Biology I. An Introduction}, vol~17. Springer

\bibitem[{Newman(2005)}]{Newman:2005}
Newman M (2005) {Power laws, Pareto distributions and Zipf's law}. Contemporary
  Physics 46 (5):323--351

\bibitem[{Olver et~al(2010)Olver, Lozier, Boisvert, and
  Clark}]{Olver:2010:NHMF}
Olver FWJ, Lozier DW, Boisvert RF, Clark CW (eds)  (2010) {NIST Handbook of
  Mathematical Functions}. Cambridge University Press, New York, NY, print
  companion to \cite{NIST:DLMF}

\bibitem[{Schiff(1999)}]{Schiff:1999}
Schiff J (1999) {The Laplace Transform: Theory and Applications}, vol~85.
  Springer

\bibitem[{Simon(1955)}]{Simon:1955}
Simon HA (1955) {On a class of skew distribution functions}. Biometrika 42
  (3-4):425--440

\bibitem[{Stein and Shakarchi(2003)}]{Stein:03}
Stein EM, Shakarchi R (2003) {Complex Analysis}. Princeton University Press

\bibitem[{Tavare(1987)}]{Tavare:1987}
Tavare S (1987) {The birth process with immigration, and the genealogical
  structure of large populations}. Journal of Mathematical Biology 25:161--168

\bibitem[{Tomasetti(2012)}]{Tomasetti:2012}
Tomasetti C (2012) {On the probability of random genetic mutations for various
  types of tumor growth}. Bulletin of Mathematical Biology 74(6):1379--1395

\bibitem[{Weisstein(2016)}]{Weisstein:2015}
Weisstein EW (2016) {``Polylogarithm." MathWorld-A Wolfram Web Resource}.
  \urlprefix\url{http. http://mathworld.wolfram.com/Polylogarithm.html}

\bibitem[{Williams et~al(2016)Williams, Werner, Barnes, Graham, and
  Sottoriva}]{Williams:2016}
Williams MJ, Werner B, Barnes CP, Graham TA, Sottoriva A (2016) {Identification
  of neutral tumor evolution across cancer types}. Nature Genetics 48:238--244

\bibitem[{Yachida et~al(2010)Yachida, Jones, Bozic, Antal, Leary, Fu, Kamiyama,
  Hruban, Eshleman, Nowak, Velculescu, Kinzler, Vogelstein, and
  Iacobuzio-Donahue}]{Yachida:2010}
Yachida S, Jones S, Bozic I, Antal T, Leary R, Fu B, Kamiyama M, Hruban RH,
  Eshleman JR, Nowak MA, Velculescu VE, Kinzler KW, Vogelstein B,
  Iacobuzio-Donahue CA (2010) {Distant metastasis occurs late during the
  genetic evolution of pancreatic cancer.} Nature 467(7319):1114--1117

\bibitem[{Zheng(1999)}]{Zheng:1999}
Zheng Q (1999) {Progress of a half century in the study of the
  Luria-Delbr\"{u}ck distribution}. Mathematical Biosciences 162(1--2):1--32

\end{thebibliography}

\end{document}